\documentclass[11pt]{amsart}
\usepackage[a4paper,margin=1in] {geometry}               

\usepackage{amsmath}
\usepackage{amsthm}
\usepackage{mathtools} 
\usepackage[algoruled]{algorithm2e}
\usepackage{enumitem}

\usepackage{tikz}
\usetikzlibrary{arrows,backgrounds,calc,fit,decorations.pathreplacing,decorations.markings,shapes.geometric}

\tikzstyle{internal} = [draw, fill, shape=circle]
\tikzstyle{external} = [shape=circle]
\tikzstyle{square}   = [draw, fill, rectangle]
\tikzstyle{triangle} = [draw, fill, regular polygon, regular polygon sides=3, inner sep=3pt]
\tikzstyle{pentagon} = [draw, fill, regular polygon, regular polygon sides=5, inner sep=2pt, minimum size=14pt]

\usepackage[margin=1cm]{caption} 
\usepackage{subfig}

\tikzset{every fit/.append style=text badly centered}

\usepackage[bookmarks=true,hypertexnames=false,pagebackref]{hyperref}
\hypersetup{colorlinks=true, citecolor=blue, linkcolor=red, urlcolor=blue}

\usepackage{cleveref}

\usepackage{todonotes}
\usepackage{mleftright}
\usetikzlibrary{positioning,chains,fit,shapes,calc}
\usetikzlibrary{trees}
\usetikzlibrary{decorations.pathmorphing}
\usetikzlibrary{decorations.markings}

\usepackage{cool}
\Style{DSymb={\mathrm d},DShorten=true,IntegrateDifferentialDSymb=\mathrm{d}}

\newcommand{\NP}{{\bf NP}}
\newcommand{\vbl}{{\sf var}}
\newcommand{\val}{\sigma}
\newcommand{\resample}{\textsc{Resample}}
\newcommand{\prs}{\textsc{Partial Rejection Sampling}}
\newcommand{\gprs}{\textsc{General Partial Rejection Sampling}}
\newcommand{\pvector}{\textbf{p}}
\newcommand{\Res}{{\sf Res}}
\newcommand{\Bad}{{\sf Bad}}

\newcommand{\PrPRS}{\Pr_{\mathrm{PRS}}}
\newcommand{\Ex}{\mathop{\mathbb{{}E}}\nolimits}
\newcommand{\rt}[2]{#1_{i,j_{i,#2}}}

\newcommand{\abs}[1]{\left\vert#1\right\vert}

\renewcommand{\Pr}{\mathop{\mathrm{Pr}}\nolimits}

\newcommand{\Cone}{6}
\newcommand{\Ctwo}{3}

\newcommand{\LLL}[0]{Lov\'asz Local Lemma}

\newtheorem{theorem}{Theorem}

\newtheorem{lemma}[theorem]{Lemma}

\newtheorem{corollary}[theorem]{Corollary}
\newtheorem*{remark}{Remark}
\newtheorem*{example}{Example}
\newtheorem{condition}[theorem]{Condition}
\newtheorem{definition}[theorem]{Definition}

\crefname{theorem}{Theorem}{Theorems}
\crefname{observation}{Observation}{Observations}
\crefname{claim}{Claim}{Claims}
\crefname{condition}{Condition}{Conditions}
\crefname{algorithm}{Algorithm}{Algorithms}
\crefname{property}{Property}{Properties}
\crefname{example}{Example}{Examples}
\crefname{fact}{Fact}{Facts}
\crefname{lemma}{Lemma}{Lemmas}
\crefname{corollary}{Corollary}{Corollaries}
\crefname{definition}{Definition}{Definitions}
\crefname{remark}{Remark}{Remarks}
\crefname{proposition}{Proposition}{Propositions}
\crefname{equation}{equation}{equations}

\makeatletter
\def\prob#1#2#3{\goodbreak\begin{list}{}{\labelwidth\z@ \itemindent-\leftmargin
                        \itemsep\z@  \topsep6\p@\@plus6\p@
                        \let\makelabel\descriptionlabel}
                      \item[\textbf{Name}]#1
                      \item[\textbf{Instance}]#2
                      \item[\textbf{Output}]#3
                \end{list}}
\makeatother

\title{Uniform Sampling through the Lov\'asz Local Lemma}

\author{Heng Guo}
\address{School of Informatics, 
University of Edinburgh, Informatics Forum, Edinburgh EH8 9AB, United Kingdom.}
\email{hguo@inf.ed.ac.uk}

\author{Mark Jerrum}
\address{School of Mathematical Sciences,
Queen Mary, University of London, Mile End Road, London E1 4NS, United Kingdom.}
\email{mj@qmul.ac.uk}

\author{Jingcheng Liu}
\address{Department of EECS, University of California, Berkeley, CA}
\email{liuexp@berkeley.edu}

\begin{document}

\begin{abstract}
  We propose a new algorithmic framework, called ``partial rejection sampling'', 
  to draw samples exactly from a product distribution, 
  conditioned on none of a number of bad  events occurring.
  Our framework builds new connections between the variable framework of the Lov\'asz Local Lemma and some classical sampling algorithms 
  such as the ``cycle-popping'' algorithm for rooted spanning trees.
  Among other applications, we discover new algorithms to sample satisfying assignments of $k$-CNF formulas with bounded variable occurrences.
\end{abstract}

\maketitle

\section{Introduction}

The Lov\'asz Local Lemma \cite{EL75} is a classical gem in combinatorics that guarantees the existence of a perfect object that avoids all events deemed to be ``bad''.
The original proof is non-constructive but there has been great progress in the algorithmic aspects of the local lemma.
After a long line of research \cite{Beck91a,Alon91,MR98,CS00,SS05,Sri08}, the celebrated result by Moser and Tardos \cite{MT10} gives efficient algorithms to find such a perfect object
under conditions that match the Lov\'asz Local Lemma in the so-called variable framework.
However, it is natural to ask whether, under the same condition, we can also sample a perfect object uniformly at random instead of merely finding one.

Roughly speaking, the resampling algorithm by Moser and Tardos \cite{MT10} works as follows.
We initialize all variables randomly. If bad events occur, then we arbitrarily choose a bad event and resample all the involved variables.
Unfortunately, it is not hard to see 
that this algorithm can produce biased samples.
This seems inevitable. As Bez{\'{a}}kov{\'{a}} et al.\ showed~\cite{BGGGS16}, sampling can be \NP-hard even under conditions that are stronger than those of the local lemma.
On the one hand, the symmetric \LLL{} only requires $ep\Delta\le 1$, where $p$ is the probability of bad events and $\Delta$ is the maximum degree of the dependency graph.
On the other hand, translating the result of \cite{BGGGS16} to this setting, one sees that as soon as $p\Delta^2\ge C$ for some constant $C$, then even approximately sampling perfect objects in the variable framework becomes \NP-hard.

The starting point of our work is a new condition (see~\Cref{cond:diamond}) under which we show that the output of the Moser-Tardos algorithm is in fact uniform (see~\Cref{thm:prs}).
Intuitively, the condition requires any two dependent bad events to be disjoint.
Indeed, instances satisfying this condition are called ``extremal'' in the study of Lov\'asz Local Lemma.
For these extremal instances, we can in fact resample in a parallel fashion, since the occurring bad events form an independent set in the dependency graph.
We call this algorithm ``partial rejection sampling'',\footnote{Despite the apparent similarity in names, our algorithm is different from ``partial resampling'' in \cite{HS13b,HS13a}. We resample \emph{all} variables in certain sets of events whereas ``partial resampling'' only resamples a subset of variables from some bad event.}
in the sense that it is like rejection sampling, but only resamples an appropriate subset of variables.

Our result puts some classical sampling algorithms under a unified framework, 
including the ``cycle-popping'' algorithm by Wilson~\cite{Wilson96} for sampling rooted spanning trees,
and the ``sink-popping'' algorithm by Cohn, Pemantle, and Propp~\cite{CPP02} for sampling sink-free orientations of an undirected graph.
Indeed, Cohn et al.\ \cite{CPP02} coined the term ``partial rejection sampling'' and asked for a general theory,
and we believe that extremal instances under the variable framework is a satisfactory answer.
With our techniques, we are able to give a new algorithm to sample solutions for a special class of $k$-CNF formulas,
under conditions matching the Lov\'asz Local Lemma (see~Corollary \ref{cor:ext-CNF}),
which is an \NP-hard task for general $k$-CNF formulas.
Furthermore, we provide \emph{explicit} formulas for the expected running time of these algorithms (see~\Cref{thm:exp-time}),
which matches the running time upper bound given by Kolipaka and Szegedy~\cite{KS11} under Shearer's condition~\cite{Shearer85}.

The next natural question is thus whether we can go beyond extremal instances.
Indeed, our main technical contribution is a general uniform sampler (\Cref{ALG:G-PRS}) that applies to \emph{any} problem under the variable framework.
The main idea is that, instead of only resampling occurring bad events, 
we resample a larger set of events so that the choices made do not block any perfect assignments in the end,
in order to make sure of uniformity in the final output.

As a simple example, we describe how our algorithm samples independent sets.
The algorithm starts by choosing each vertex with probability $1/2$ independently. 
At each subsequent round, in the induced subgraph on the currently chosen vertices, the algorithm finds all the connected components of size $\ge 2$.
Then it resamples all these vertices and their boundaries (which are unoccupied). 
And it repeats this process until there is no edge with both endpoints occupied.
What seems surprising is that this simple process does yield a uniformly random independent set when it stops.
Indeed, as we will show in~\Cref{thm:hard-core}, this simple process is an exact sampler for \emph{weighted independent sets} (also known as the \emph{hard-core model} in statistical physics).
In addition, it runs in expected linear time under a condition that matches, up to a constant factor, the \emph{uniqueness threshold} of the model (beyond which the problem of approximate sampling becomes \NP-hard).

In the more general setting, we will choose the set of events to be resampled, denoted by $\Res$,
iteratively.
We start from the set of occurring bad events.
Then we include all neighbouring events of the current set $\Res$, until there is no event $A$ on the boundary of $\Res$
such that the current assignment, projected on the common variables of $A$ and $\Res$, can be extended so that $A$ may happen.
In the worst case, we will resample all events (there is no event in the boundary at all).
In that scenario the algorithm is the same as a naive rejection sampling,
but typically we resample fewer variables in every step.
We show that this is a uniform sampler on assignments that avoid all bad events once it stops (see~\Cref{thm:g-prs}).

One interesting feature of our algorithm is that, unlike Markov chain based algorithms,
ours does not require the solution space (or any augmented space) to be connected.
Moreover, our sampler is exact;
that is, when the algorithm halts, the final distribution is precisely the desired distribution.
Prior to our work, most exact sampling algorithms were obtained by coupling from the past \cite{PW96}.
We also note that previous work on the Moser-Tardos output distribution, such as \cite{HSS11}, 
is not strong enough to guarantee a uniform sample (or $\varepsilon$-close to uniform in terms of total variation distances).

We give sufficient conditions that guarantee a linear expected running time of our algorithm in the general setting (see \Cref{thm:exp-time:G-prs}).
The first condition is that $p\Delta^2$ is bounded above by a constant.
This is optimal up to constants in observance of the \NP-hardness result in \cite{BGGGS16}.
Unfortunately, the condition on $p\Delta^2$ alone does not make the algorithm efficient.
In addition, we also need to bound the expansion from bad events to resampling events,
which leads to an extra condition on intersections of bad events.
Removing this extra condition seems to require substantial changes to our current algorithm.

To illustrate the result, we apply our algorithm to sample satisfying assignments of $k$-CNF formulas 
in which the degree of each variable (the number of clauses containing it) is at most $d$.
We say that a $k$-CNF formula has intersection $s$ if any two \emph{dependent} clauses share at least $s$ variables.
The extra condition from our analysis naturally leads to a lower bound on $s$.
Let $n$ be the number of variables.
We (informally) summarize our results on $k$-CNF formulas as follows (see Corollary \ref{cor:sharing-CNF} and \Cref{thm:sharing-CNF-hard}):
\begin{itemize}
	\item If $ d\le \frac{1}{6e}\cdot 2^{k/2}$, $dk \ge 2^{3e}$ and $s\ge \min\{\log_2 dk,k/2\}$,
    then the general partial rejection resampling algorithm outputs a uniformly random solution to a $k$-CNF formula with degree $d$ and intersection $s$
    in expected running time $O(n)$.
  \item If $d\ge 4\cdot 2^{k/2}$ (for an even $k$), then even if $s=k/2$, 
	  it is \NP-hard even to \emph{approximately} sample a solution to a $k$-CNF formula with degree $d$ and intersection $s$.
\end{itemize}
As shown in the hardness result, the intersection bound does not render the problem trivial.

Previously, sampling/counting satisfying assignments of $k$-CNF formulas required the formula to be monotone and $d\le k$ to be large enough \cite{BGGGS16} (see also \cite{BDK06,LL15}).
Although our result requires an additional lower bound on intersections,
not only does it improve the dependency of $k$ and $d$ \emph{exponentially}, but also achieves a \emph{matching} constant $1/2$ in the exponent.
Furthermore the samples produced are exactly uniform.
Thus, if the extra condition on intersections can be removed,
we will have a sharp phase transition at around $d=O(2^{k/2})$ in the computational complexity of 
sampling solutions to $k$-CNF formulas with bounded variable occurrences.
A similar sharp transition has been recently established for, e.g., sampling configurations in the hard-core model~\cite{Wei06,SS14,GSV16}.

Simultaneous to our work, Hermon, Sly, and Zhang \cite{HSZ16} showed that Markov chains for monotone $k$-CNF formulas are rapidly mixing, if $d\le c2^{k/2}$ for a constant $c$.
In another parallel work, Moitra~\cite{Moi16} gave a novel algorithm to sample solutions for general $k$-CNF when $d\lesssim 2^{k/60}$. 
We note that neither results are directly comparable to ours and the techniques are very different.
Both of these two samplers are approximate while ours is exact.
Moreover, ours does not require monotonicity (unlike \cite{HSZ16}),
and allows larger $d$ than \cite{Moi16} but at the cost of an extra intersection lower bound.
Unfortunately, our algorithm can be exponentially slow when the intersection $s$ is not large enough.
In sharp contrast, as shown by Hermon et al.\ \cite{HSZ16}, Markov chains mix rapidly for $d\le c2^{k}/k^2$ when $s=1$.

While the study of algorithmic \LLL{} has progressed beyond the variable framework~\cite{HS14a,AI16,HV15}, we restrict our focus to the variable framework in this work.
It is also an interesting future direction to investigate and extend our techniques of uniform sampling beyond the variable framework.
For example, one may want to sample a permutation that avoids certain patterns.
The classical sampling problem of perfect matchings in a bipartite graph can be formulated in this way.

Since the conference version of this paper appeared \cite{GJL17},
a number of applications of the partial rejection sampling method have been found \cite{GJ18a,GJ18b,GJ18c}.
One highlight is the first fully polynomial-time randomised approximation scheme (FPRAS) for all-terminal network reliability \cite{GJ18a}.
For the extremal instances, tight running time bounds have also been obtained \cite{GH18}.
Moreover, partial rejection sampling is adapted to dynamic and distributed settings as well \cite{FLY18}.

\section{Partial Rejection Sampling}\label{sec:PRS}

We first describe the ``variable'' framework.
Let $\{X_1,\dots,X_n\}$ be a set of random variables.
Each $X_i$ can have its own distribution and range $D_i$.
Let $\{A_1,\dots,A_m\}$ be a set of ``bad'' events that depend on $X_i$'s.
For example, for a constraint satisfaction problem (CSP) with variables $X_i$ ($1\le i\le n$) and constraints $C_j$ ($1\le j\le m$), 
each $A_j$ is the set of unsatisfying assignments of $C_j$ for $1\le j\le m$.
Let $\vbl(A_i)$ be the (index) set of variables that $A_i$ depends on.

The dependency graph $G=(V,E)$ has $m$ vertices, 
identified with the integers $\{1,2,\ldots,m\}$,
corresponding to the events $A_i$, 
and $(i,j)$ is an edge if $A_i$ and $A_j$ depend on one or more common variables, and $i \neq j$.
In other words, for any distinct $i,j$, $(i,j)\in E$ if $\vbl(A_i)\cap\vbl(A_j)\neq\emptyset$.
We write $A_i\sim A_j$ if the vertices $i$ and~$j$ are adjacent in~$G$.
The asymmetric Lov\'{a}sz Local Lemma \cite{EL75} states the following.

\begin{theorem}  \label{thm:LLL}
	If there exists a vector $\boldsymbol x \in [0,1)^m$ such that $\forall i \in [m]$,
  \begin{align} \label{eqn:LLL}
    \Pr(A_i)\le x_i \prod_{(i,j)\in E}(1-x_j), 
  \end{align}
  then $\displaystyle\Pr\left(\bigwedge_{i=1}^m \overline{A_i}\right)\ge \prod_{i=1}^m (1-x_i)>0$.
\end{theorem}

Theorem \ref{thm:LLL} has a condition that is easy to verify, but not necessarily optimal.
Shearer \cite{Shearer85} gave the optimal condition for the local lemma to hold for a fixed dependency graph $G$.
To state Shearer's condition, we will need the following definitions.
Let $p_i:=\Pr(A_i)$ for all $1\le i\le m$.
Let $\mathcal{I}$ be the collection of independent sets of $G$.
Define the following quantity:
\begin{align*}
  q_I(\pvector):=\sum_{J\in\mathcal{I},\, I\subseteq J}(-1)^{|J|-|I|}\prod_{i\in J}p_i,
\end{align*}
where $\pvector=(p_1,\ldots,p_m)$.
When there is no confusion we also simply write $q_I$ instead of the more cumbersome $q_I(\pvector)$.
Moreover, to simplify the notation, let $q_i:=q_{\{i\}}$ for $1\le i\le m$.
Note that if $I\notin \mathcal{I}$, $q_I=0$.

\begin{theorem}[Shearer \cite{Shearer85}]
  If $q_I\ge 0$ for all $I\subseteq V$,
  then $\Pr\left(\bigwedge_{i=1}^m \overline{A_i}\,\right)\ge q_{\emptyset}$.
  \label{thm:Shearer}
\end{theorem}

In particular, if the condition holds with $q_{\emptyset}>0$, then $\Pr\left(\bigwedge_{i=1}^m \overline{A_i}\,\right)>0$.

Neither Theorem \ref{thm:LLL} nor Theorem \ref{thm:Shearer} yields an efficient algorithm to find the assignment avoiding all bad events,
since they only guarantee an exponentially small probability. 
There has been a long line of research devoted to an algorithmic version of LLL,
culminating in Moser and Tardos \cite{MT10} with essentially the same condition as in Theorem \ref{thm:LLL}.
The \resample\ algorithm of Moser and Tardos is very simple, described in Algorithm \ref{alg:Moser-Tardos}.

\begin{algorithm}
  \caption{The \resample\ algorithm}
  \label{alg:Moser-Tardos}
  \begin{enumerate}[itemsep=0.5em, topsep=0.5em]
    \item Draw independent samples of all variables $X_1,\dots,X_n$ from their respective distributions.
    \item While at least one $A_i$ holds, pick one such $A_i$ arbitrarily and resample all variables in $\vbl(A_i)$.
    \item Output the current assignment.
  \end{enumerate}
\end{algorithm}

In \cite{MT10}, Moser and Tardos showed that Algorithm \ref{alg:Moser-Tardos} finds a good assignment very efficiently.

\begin{theorem}[Moser and Tardos \cite{MT10}]
  Under the condition of Theorem \ref{thm:LLL},
  the expected number of resampling steps in~\Cref{alg:Moser-Tardos} is at most $\sum_{i=1}^m\frac{x_i}{1-x_i}$.
  \label{thm:Moser-Tardos}
\end{theorem}

Unfortunately, the final output of Algorithm \ref{alg:Moser-Tardos} is not distributed 
as we would like, namely as a product distribution conditioned on avoiding all bad events.

In addition, Kolipaka and Szegedy~\cite{KS11} showed that up to Shearer's condition,
Algorithm \ref{alg:Moser-Tardos} is efficient.
Recall that $q_i:=q_{\{i\}}$ for $1\le i\le m$.

\begin{theorem}[Kolipaka and Szegedy~\cite{KS11}]
  If $q_I\ge 0$ for all $I\in\mathcal{I}$ and $q_{\emptyset}>0$,
  then the expected number of resampling steps in~\Cref{alg:Moser-Tardos} is at most $\sum_{i=1}^m\frac{q_i}{q_{\emptyset}}$.
  \label{thm:KS11}
\end{theorem}

On the other hand, Wilson's cycle-popping algorithm \cite{Wilson96} is very similar to the \resample\ algorithm but it outputs a uniformly random rooted spanning tree.
Another similar algorithm is the sink-popping algorithm by Cohn, Pemantle, and Propp \cite{CPP02} to generate a sink-free orientation uniformly at random.
Upon close examination of these two algorithms, we found a common feature of both problems.

\begin{condition}
  If $(i,j)\in E$ $($or equivalently $A_i\sim A_j)$, then $\Pr(A_i\wedge A_j)=0$; 
  namely the two events $A_i$ and $A_j$ are disjoint if they are dependent.
  \label{cond:diamond}
\end{condition}

In other words, any two events $A_i$ and $A_j$ are either independent or disjoint.
These instances have been noticed in the study of \LLL{}.
They are the ones that minimize $\Pr\left(\bigwedge_{i=1}^m \overline{A_i}\,\right)$ given Shearer's condition
(namely $\Pr\left(\bigwedge_{i=1}^m \overline{A_i}\,\right) = q_{\emptyset}$).
Instances satisfying Condition \ref{cond:diamond} have been named \emph{extremal} \cite{KS11}.

We will show that, given Condition \ref{cond:diamond},
the final output of the \resample\ algorithm is a sample from 
a conditional product distribution (Theorem \ref{thm:prs}).
Moreover, we will show that under Condition \ref{cond:diamond},
the running time upper bound $\sum_{i=1}^m\frac{q_i}{q_{\emptyset}}$ given by Kolipaka and Szegedy (Theorem \ref{thm:KS11}) is indeed the exact expected running time.
See Theorem \ref{thm:exp-time}.

In fact, when Condition \ref{cond:diamond} holds,
at each step of Algorithm \ref{alg:Moser-Tardos}, the 
occurring events form an independent set of the dependency graph~$G$.
Think of the execution of Algorithm \ref{alg:Moser-Tardos} as going in rounds.
In each round we find the set~$I$ of bad events that occur.
Due to Condition \ref{cond:diamond}, $\vbl(A_i)\cap\vbl(A_j)=\emptyset$ for any $i,j\in I$,
i.e., $I$ is an independent set in the dependency graph.
Therefore, we can resample all variables involved in the occurring 
bad events without interfering with each other.
This motivates Algorithm \ref{ALG:PRS}.

We call Algorithm \ref{ALG:PRS} the \prs\ algorithm.
This name was coined by Cohn, Pemantle, and Propp \cite{CPP02}.
Indeed, they ask as an open problem how to generalize their sink-popping algorithm and Wilson's cycle popping algorithm.
We answer this question under the variable framework.
\prs\ differs from the normal rejection sampling algorithm by only resampling ``bad'' events.
Moreover, \Cref{ALG:PRS} is uniform only on extremal instances, and is a special case of \Cref{ALG:G-PRS} given in \Cref{sec:G-prs},
which is a uniform sampler for all instances.

\begin{algorithm}
  \caption{\prs\ for extremal instances}
  \label{ALG:PRS}
  \begin{enumerate}[itemsep=0.5em, topsep=0.5em]
    \item Draw independent samples of all variables $X_1,\dots,X_n$ from their respective distributions.
    \item While at least one bad event holds, find the independent set $I$ of occurring $A_i$'s.
      Independently resample all variables in $\bigcup_{i\in I}\vbl(A_i)$.
    \item Output the current assignment.
  \end{enumerate}
\end{algorithm}

In fact, Algorithm \ref{ALG:PRS} is the same as the parallel version of \Cref{alg:Moser-Tardos} by Moser and Tardos~\cite{MT10} for extremal instances.
Suppose each event is assigned to a processor, which determines whether the event holds by looking at the variables associated with the event.
If the event holds then all associated variables are resampled.
No conflict will be created due to Condition \ref{cond:diamond}.

In the following analysis, we will use the resampling table idea,
which has appeared in both the analysis of Moser and Tardos \cite{MT10} and Wilson \cite{Wilson96}.
Note that we only use this idea to analyze the algorithm rather than to really create the table in the execution.
Associate each variable $X_i$ with an infinite stack of random values $\{X_{i,1},X_{i,2},\dots\}$.
This forms the \emph{resampling table} where each row represents a variable and there are infinitely many columns,
as shown in Table \ref{tab:resampling}.
In the execution of the algorithm, when a variable needs to be resampled, 
the algorithm draws the top value from the stack, or equivalently moves from the current entry in the resampling table to its right.

\begin{table}  
  \caption{A resampling table with $4$ variables}
  \label{tab:resampling}
  \centering
  \begin{tabular}{|r|cccc|}
    \hline
    $X_1$ & $X_{1,1}$ & $X_{1,2}$ & $X_{1,3}$ & $\dots$\\
    \hline
    $X_2$ & $X_{2,1}$ & $X_{2,2}$ & $X_{2,3}$ & $\dots$\\
    \hline
    $X_3$ & $X_{3,1}$ & $X_{3,2}$ & $X_{3,3}$ & $\dots$\\
    \hline
    $X_4$ & $X_{4,1}$ & $X_{4,2}$ & $X_{4,3}$ & $\dots$\\
    \hline
  \end{tabular}
\end{table}

Let $t$ be a positive integer to denote the round of Algorithm \ref{ALG:PRS}.
Let $j_{i,t}$ be the index of the variable $X_i$ in the resampling table at round $t$.
In other words, at the $t$-th round, $X_i$ takes value $\rt{X}{t}$.
Thus, the set $\val_t=\{\rt{X}{t}\mid 1\le i\le n\}$ is the current assignment of variables at round $t$.
This $\val_t$ determines which events happen.
Call the set of occurring events, viewed as a subset of the vertex set of the dependency graph, $I_t$.
(For convenience, we shall sometimes identify the event $A_i$ with its index~$i$;  
thus, we shall refer to the ``events in $S$'' rather than the ``events indexed by $S$''.)
As explained above, $I_t$ is an independent set of~$G$ due to Condition \ref{cond:diamond}.
Then variables involved in any of the events in~$I_t$ are resampled.
In other words,
\begin{align*}
  j_{i,t+1}=
  \begin{cases}
    j_{i,t}+1 & \text{if }\exists \ell \in I_t \text{ such that } i\in\vbl(A_\ell);\\
    j_{i,t} & \text{otherwise.}
  \end{cases}
\end{align*}
Therefore, any event that happens in round $t+1$,
must share at least one variable with some event in~$I_t$ (possibly itself).
In other words, $I_{t+1}\subseteq \Gamma^{+}(I_t)$ where $\Gamma^+(\cdot)$ denotes the set of all neighbours of~$I$ unioned with $I$ itself.
This inspires the notion of independent set sequences (first introduced in \cite{KS11}).

\begin{definition}  \label{def:IS-seq}
A list $\mathcal{S}=S_1,S_2,\dots,S_\ell$ of independent sets in $G$ is called an \emph{independent set sequence} 
  if $S_i\neq\emptyset$ for all $1\le i\le \ell-1$ and for every $1\le i\le \ell-1$, $S_{i+1}\subseteq \Gamma^+(S_i)$.
\end{definition}

We adopt the convention that 
the empty list is an independent set sequence with $\ell=0$.
Note that we allow $S_{\ell}$ to be $\emptyset$.

Let $M$ be a resampling table.
Suppose running Algorithm \ref{ALG:PRS} on $M$ does not terminate up to some integer $\ell\ge 1$ rounds.
Define the \emph{log} of running Algorithm \ref{ALG:PRS} on $M$ up to round $\ell$ as the sequence of independent sets $I_1,I_2,\dots,I_{\ell}$ created by this run.
Thus, for any $M$ and $\ell\ge 1$, the log $I_1,I_2,\dots,I_{\ell}$ must be an independent set sequence.
Moreover, if \Cref{ALG:PRS} terminates at round $T$,
let $\sigma_{t}:=\sigma_{T}$ if $t>T$.
Denote by $\mu(\cdot)$ the product distribution of all random variables.

\begin{lemma}
  Suppose Condition \ref{cond:diamond} holds.
  Given any log $\mathcal{S}=S_1,S_2,\dots,S_\ell$ of length $\ell\ge 1$ and 
  conditioned on seeing the log $\mathcal{S}$,
  $\val_{\ell+1}$ is a random sample from the product distribution
  conditioned on the event $\bigwedge_{i \in [m] \setminus  \Gamma^{+} (S_\ell)} \overline{A_i}$,
  namely from $\mu\big(\cdot \mid \bigwedge_{i \in [m] \setminus  \Gamma^{+} (S_\ell)} \overline{A_i}\big)$.
  \label{lem:diamond}
\end{lemma}

We remark that Lemma~\ref{lem:diamond} is not true for non-extremal instances (that is, if \Cref{cond:diamond} fails).
In particular, Lemma~\ref{lem:diamond} says that given any log, every valid assignment is not only reachable, but also with the correct probability.
This is no longer the case for non-extremal instances --- some valid assignments from the desired conditional product distribution could be ``blocked'' under the log $\mathcal{S}$.
In \Cref{sec:G-prs} we show how to instead achieve uniformity by resampling an ``unblocking'' set of bad events.

\begin{proof}
  The set of occurring events at round~$\ell$ is $S_\ell$.
  Hence $\val_{\ell+1}$ does not make any of the $A_i$'s happen where $i\notin \Gamma^{+}(S_\ell)$.
  Call an assignment $\val$ \emph{valid} if none of the $A_i$'s happen where $i\notin \Gamma^{+}(S_\ell)$.
  To show that $\val_{\ell+1}$ has the desired conditional product distribution,
  we will show that the probabilities of getting any two valid assignments $\val$ and $\val'$ are 
  proportional to their probabilities of occurrence in $\mu(\cdot)$.

  Let $M$ be the resampling table so that the log of the algorithm is $\mathcal{S}$ up to round $\ell\ge 1$, and $\val_{\ell+1}=\val$.
  Indeed, since we only care about events up to round $\ell+1$, we may truncate the table so that $M=\{X_{i,j}\;\vert\; 1\le i\le n,\;\; 1\le j\le j_{i,\ell+1}\}$.
  Let $M'=\{X_{i,j}'\;\vert\; 1\le i\le n,\;\; 1\le j\le j_{i,\ell+1}\}$ be another table where $X_{i,j}'=X_{i,j}$ if $j< j_{i,\ell+1}$ for any $i\in[n]$, but $\val_{\ell+1}=\val'$.
  In other words, we only change the values in the final round $(\rt{X}{\ell+1})$, and 
  only to another valid assignment.

  We claim that the algorithm running on $M'$ generates the same log $\mathcal{S}$.
  The lemma then follows by the following argument.
  Assuming the claim holds, then conditioned on the log $\mathcal{S}$, 
  every possible table $M$ such that $\val_{\ell+1}=\val$ is one-to-one corresponding to another table $M'$ so that $\val_{\ell+1}=\val'$.
  It implies that for every pair of valid assignments $\val,\val'$, 
  there is a bijection between the resampling tables resulting in them.
  The ratio between the probability of two corresponding tables is exactly the ratio between the probabilities of $\val$ and $\val'$ under $\mu(\cdot)$.
  Since the probability of getting a particular~$\val$ in round $\ell+1$ is the sum over the probabilities of all resampling tables (conditioned on the log $\mathcal{S}$) leading to $\val$, 
  the probability of getting $\val$ is also proportional to its weight under $\mu(\cdot)$.

  Suppose the claim fails and the logs obtained by running the algorithm on $M$ and~$M'$ differ.
  Let $t_0\le \ell$ be the first round where resampling changes.
  Without loss of generality, let $A$ be an event that occurs in $S_{t_0}$ on $M'$ but not on $M$.
  Moreover, there must be a non-empty set of variables $Y \subseteq \vbl(A)$ that have values $(\rt{X}{\ell+1})$,
  as otherwise the two runs would be identical.
  Since resampling history does not change before $t_0$, in the $M'$ run, the assignment of variables in $Y$ must be $(\rt{X'}{\ell+1})$ at time $t_0$.

  We claim that $Y=\vbl(A)$.
  If the claim does not hold, then $Z:=\vbl(A)\setminus Y\neq\emptyset$.
  Any variable in $Z$ has not reached final round, and must be resampled in the $M$ run.
  Let $X_j\in Z$ be the first such variable being resampled at or after round $t_0$ in the $M$ run.
  (The choice of $X_j$ may not be unique, and we just choose an arbitrary one.)
  Recall that $Y \neq \emptyset$, $A$ can no longer happen, thus there must be $A' \neq A$ causing such a resampling of $X_j$.
  Then $\vbl(A)\cap \vbl(A')\neq\emptyset$.
  Consider any variable $X_k$ with $k\in \vbl(A)\cap \vbl(A')$.
  It is resampled at or after time $t_0$ in the $M$ run due to $A'$.
  Hence $X_k\in Z$ for any such $k$.
  Moreover, in the $M$ run, until $A'$ happens, $X_k$ has not been resampled since time $t_0$,
  because $A'$ is the first resampling event at or after time $t_0$ that involves variables in $Z$.
  On the other hand, in the $M'$ run, $X_k$'s value causes $A$ to happen at time $t_0$.
  Hence, there exists an assignment on variables in $\vbl(A)\cap \vbl(A')$ such that both $A$ and $A'$ happen.
  Clearly this assignment can be extended to a full assignment so that both $A$ and $A'$ happen.
  However, $A\sim A'$ as they share the variable $X_{j}$.
  Due to Condition \ref{cond:diamond}, $A\cap A'=\emptyset$. Contradiction!
  Therefore the claim holds.

  We argue that the remaining case, $Y=\vbl(A)$, is also not possible.
  Since $A$ occurs in the $M'$~run, we know, by the definition of $\sigma'$, 
  that $A\in\Gamma^+(S_\ell)$.  Thus, some event whose variables intersect with 
  those in~$A$ must occur in the $M$ run.  But when the algorithm attempts to update 
  variables shared by these two events in the $M$ run, it will access values beyond the final round of the
  resampling table, a contradiction.
\end{proof}

\begin{theorem}
  When Condition \ref{cond:diamond} holds and Algorithm \ref{ALG:PRS} halts,
  its output is the product distribution $\mu(\cdot)$ conditioned on avoiding all bad events.
  \label{thm:prs}
\end{theorem}

\begin{proof}
%
  Let an independent set sequence $\mathcal{S}$ of length $\ell$ be the log of any successful run. 
  Then  $S_{\ell}=\emptyset$.
  By Lemma \ref{lem:diamond}, conditioned on the log $\mathcal{S}$, the output assignment $\sigma$ 
  is $\mu\big(\cdot \mid \bigwedge_{i \in [m] \setminus  \Gamma^{+} (S_\ell)} \overline{A_i}\big)
  = \mu\big(\cdot \mid \bigwedge_{i \in [m]} \overline{A_i}\big)$.
  This is valid for any possible log,
  and the theorem follows.
\end{proof}

In other words, let $\Sigma$ be the set of assignments that avoid all bad events.
Let $U$ be the output of Algorithm \ref{ALG:PRS}.  In the case that all variables are
sampled from the uniform distribution, 
we have $\Pr(U=\sigma)=\frac{1}{|\Sigma|}$, for all $\sigma\in\Sigma$.

\section{Expected running time of Algorithm \ref{ALG:PRS}}

In this section we give an explicit formula for the running time of Algorithm \ref{ALG:PRS}.
We assume that \Cref{cond:diamond} holds throughout the section.

We first give a combinatorial explanation of $q_I$ for any independent set $I$ of the dependency graph $G$.
To simplify the notation, we denote the event $\bigwedge_{i\in S}A_i$, i.e.,
the conjunction of all events in~$S$, by $A(S)$.  

For any set~$I$ in the dependency graph, 
we denote by $p_I$ the probability $\Pr_\mu(A(I))$ that all events in~$I$ happen
(and possibly some other events too).
If $I$ is an independent set in the dependency graph, any two events in $I$ are independent and
\begin{align}
  p_I=\prod_{i\in I}p_i.
  \label{eqn:pi}
\end{align}
Moreover, for any set $J$ of events that is not an independent set, 
we have $p_J=0$ due to~\Cref{cond:diamond}.

On the other hand, the quantity $q_I$ is in fact the probability 
that exactly the events in $I$ happen and no others do.
This can be verified using inclusion-exclusion, together with Condition \ref{cond:diamond}:
\begin{align}
  \Pr_\mu\left(\bigwedge_{i\in I}A_i\wedge\bigwedge_{i\notin I}\overline{A_i}\right) & = \sum_{J\supseteq I}(-1)^{|J\setminus I|}p_J \notag\\
  & = \sum_{J\in\mathcal{I},\ J\supseteq I}(-1)^{|J\setminus I|}p_J = q_I,
  \label{eqn:qi}
\end{align}
where $\mathcal{I}$ denotes the collection of all independent sets of~$G$.
Since the events $(\bigwedge_{i\in I}A_i\wedge\bigwedge_{i\notin I}\overline{A_i})$ are mutually exclusive for different $I$'s,
\begin{align*}
  \sum_{I\in\mathcal{I}} q_I=1.
\end{align*}
Moreover, since the event $A(I)$ is the union over $J\supseteq I$ of the events 
$(\bigwedge_{i\in J}A_i\wedge\bigwedge_{i\notin J}\overline{A_i})$, we have
\begin{align}
  p_I=\sum_{J\in\mathcal{I},\ J\supseteq I} q_J.
  \label{eqn:pi-qi}
\end{align}

\begin{lemma}\label{lem:log-prob}
  Assume \Cref{cond:diamond} holds.  
  Let $\mathcal{S}=S_1,\ldots,S_\ell$ be an independent set sequence of length $\ell>0$.
  Then in \Cref{ALG:PRS},
  \begin{align*}
    \Pr\big(\text{the log is $\mathcal{S}$ up to round $\ell$}\big)=q_{S_{\ell}}\prod_{t=1}^{\ell-1}p_{S_t}.
  \end{align*}
\end{lemma}
\begin{proof}
  Clearly, if $q_{S_{\ell}}=0$, then the said sequence will never happen.
  We assume that $q_{S_{\ell}}>0$ in the following.

  Recall that $\mu$ is the product distribution of sampling all variables independently.
  We need to distinguish the probability space with respect to $\mu$ from that with respect to the execution of the algorithm.
  We write $\PrPRS(\cdot)$ to refer to the algorithm,
  and write $\Pr_\mu(\cdot)$ to refer to the (static) space with respect to~$\mu$.
  As noted before, to simplify the notation we will use $A(S)$ 
  to denote the event $\bigwedge_{i\in S} A_i$, where $S\subseteq[m]$.
  In addition, $B(S)$ will be used to denote $\bigwedge_{i\in S} \overline{A_i}$.
  For $I\in\mathcal{I}$, define 
  $$\partial I:=\Gamma^+(I)\setminus I,\quad I^e:=[m]\setminus\Gamma^+(I),\quad
  \text{and}\quad I^c:=[m]\setminus I=\partial I\cup I^e.$$
  So $\partial I$ is the ``boundary'' of $I$,
  comprising events that are not in~$I$ but which depend on~$I$, and $I^e$ is the ``exterior''
  of $I$, comprising events that are independent of all events in~$I$.  The complement $I^c$
  is simply the set of all events not in~$I$.  Note that $B(I^c)=B(\partial I)\wedge B(I^e)$.
  As examples of the notation, 
  $\Pr_\mu(A(I))=\prod_{i\in I}p_i=p_I$ is the probability that all events in~$I$
  occur under~$\mu$, and 
  $\Pr_\mu(A(I)\wedge B(I^c))=q_I$ is the probability that exactly the events in~$I$
  occur. 
  
  By the definition of $I^e$, we have that
  \begin{align}
    \Pr_\mu\left(B(I^e)\mid A(I)\right)=\Pr_\mu(B(I^e)),
    \label{eqn:not-BI-cond-I}
  \end{align}
  and, by Condition \ref{cond:diamond}, that
  \begin{align}
    A(I)\wedge B(\partial I)=A(I).
    \label{eqn:BI-wedge-I}
  \end{align}
  Hence for any $I\in\mathcal{I}$,
  \begin{align}
    q_{I} & = \Pr_\mu\big(A(I)\wedge B(I^c)\big)\notag\\
    & = \Pr_\mu\big(A(I)\wedge B(\partial I) \wedge B(I^e)\big) \notag\\
    & = \Pr_\mu\big(A(I)\wedge B(I^e)\big) &&\text{by \eqref{eqn:BI-wedge-I}} \notag\\
    & = \Pr_\mu\left(A(I)\right)\Pr_\mu\left(B(I^e)\right). &&\text{by \eqref{eqn:not-BI-cond-I}}\label{eqn:q_I}    
  \end{align}

  We prove the lemma by induction. 
  It clearly holds when $\ell=1$.
  At round $\ell\ge 2$, since we only resample variables that are involved in $S_{\ell-1}$,
  we have that $S_{\ell}\subseteq\Gamma^+(S_{\ell-1})$.
  Moreover, variables are not resampled in any $A_i$ where $i\in S_{\ell-1}^e$, and hence
  \begin{align}\label{eqn:St+1}
    B(S_{\ell}^c) \wedge B(S_{\ell-1}^e) = B(S_{\ell}^c).
  \end{align}
  Conditioned on $S_{\ell-1}$, by Lemma \ref{lem:diamond},
  the distribution of $\val_{\ell}$ at round $\ell$ is the product distribution conditioned 
  on none of the events outside of $\Gamma^+(S_{\ell-1})$ occuring; 
  namely, it is $\Pr_\mu\big( \cdot\mid B(S_{\ell-1}^e) \big)$.
  Thus the probability of getting $S_{\ell}$ in round $\ell$ is
  \begin{align}
    &\PrPRS\big(A(S_{\ell})\wedge B(S_{\ell}^c)\text{ holds in round $\ell$}\bigm|
       \text{prior log is }S_1,\ldots,S_{\ell-1}\big)\notag\\
    &\qquad{}=\Pr_{\mu}\big( A(S_{\ell})\wedge B(S_{\ell}^c)\mid B(S_{\ell-1}^e) \big) \notag \\
    &\qquad{} = \frac{\Pr_\mu\big(A(S_{\ell})\wedge B(S_{\ell}^c)\wedge B(S_{\ell-1}^e)\big)}{\Pr_\mu(B(S_{\ell-1}^e))} \notag\\
    &\qquad{} = \frac{\Pr_\mu\big(A(S_{\ell})\wedge B(S_{\ell}^c)\big)}{\Pr_\mu(B(S_{\ell-1}^e))} 
       &&\text{by \eqref{eqn:St+1}}\notag\\
    &\qquad{} = \frac{q_{S_{\ell}}}{\Pr_\mu(B(S_{\ell-1}^e))}.
    \label{eqn:rec-Si}
  \end{align}
  By \eqref{eqn:rec-Si} and the induction hypothesis, we have
  \begin{align*}
    \PrPRS\left( \text{the log is $\mathcal{S}$ up to round $\ell$} \right) 
    & = \frac{q_{S_{\ell}}}{\Pr_\mu(B(S_{\ell-1}^e))}\cdot q_{S_{\ell-1}}\prod_{t=1}^{\ell-2}p_{S_t}\\
    & = q_{S_{\ell}}\prod_{t=1}^{\ell-1}p_{S_t},
  \end{align*}
  where to get the last line we used \eqref{eqn:q_I} on $S_{\ell-1}$.
\end{proof}

Essentially the proof above is a delayed revelation argument.
At each round $1\le t\le \ell-1$, we only reveal variables that are involved in $S_t$.
Thus, at round $\ell$, we have revealed all variables that are involved in $\mathcal{S}$.
With respect to these variables, the sequence $\mathcal{S}$ happens with probability $p_{\mathcal{S}}$. 
Condition \ref{cond:diamond} guarantees that what we have revealed so far does not interfere with the final output (cf.\ Lemma \ref{lem:diamond}).
Hence the final state happens with probability $q_{S_{\ell}}$.

We write $p_{\mathcal{S}}=\prod_{i=1}^{\ell}p_{S_i}$ for an independent set sequence $\mathcal{S}$ of length $\ell\ge 0$.
Note the convention that $p_{\mathcal{S}}=1$ if $\mathcal{S}$ is empty and $\ell=0$.
Then, Lemma \ref{lem:log-prob} implies the following equality,
which is first shown by Kolipaka and Szegedy \cite{KS11} in the more general (not necessarily extremal) setting of the local lemma.

\begin{corollary}
  Assume \Cref{cond:diamond} holds.
  If $q_{\emptyset}>0$, then
  \begin{align*}
    \sum_{\mathcal{S}\text{ s.t.\ }S_1=I} p_{\mathcal{S}}q_{\emptyset}=q_{I},
  \end{align*}
  where $\mathcal{S}$ is an independent set sequence and $I$ is an independent set of $G$.
  \label{cor:pi-seq}
\end{corollary}
\begin{proof}
  First we claim that if $q_{\emptyset}>0$, then Algorithm \ref{ALG:PRS} halts with probability $1$.
  Conditioned on any log $\mathcal{S}= S_1,\dots,S_{\ell-1}$, by Lemma \ref{lem:diamond},
  the distribution of $\val_{\ell}$ at round $\ell$ is $\mu\big( \cdot\mid B(S_{\ell-1}^e) \big)$.
  The probability of getting a desired assignment is thus $\mu\big( B([m]) \mid B(S_{\ell-1}^e) \big)=\frac{\mu(B([m]))}{\mu(B(S_{\ell-1}^e))}\ge\mu(B([m]))=q_{\emptyset}$.
  Hence the probability that the algorithm does not halt at time $t$ is at most $(1-q_{\emptyset})^t$,
  which goes to $0$ as $t$ goes to infinity.

  Then we apply Lemma \ref{lem:log-prob} when $\emptyset$ is the final independent set.
  The left hand side is the total probability of all possible halting logs 
  whose first independent set is exactly~$I$.
  This is equivalent to getting exactly $I$ in the first step, which happens with probability $q_{I}$.
\end{proof}

As a sanity check, the probability of all possible logs should sum to $1$ when $q_{\emptyset}>0$ and the algorithm halts with probability $1$.
Indeed, by Corollary \ref{cor:pi-seq},
\begin{align*}
  \sum_{\mathcal{S}}p_{\mathcal{S}}q_{\emptyset}&=\sum_{I\in\mathcal{I}}
  \;\;\sum_{\mathcal{S}\text{ s.t.\ }S_1=I}p_{\mathcal{S}}q_{\emptyset}=\sum_{I\in\mathcal{I}}q_I=1,
\end{align*}
where $\mathcal{S}$ is an independent set sequence.
In other words,
\begin{align}\label{eqn:sum-traces}
  \sum_{\mathcal{S}}p_{\mathcal{S}}=\frac{1}{q_{\emptyset}},
\end{align}
where $\mathcal{S}$ is an independent set sequence.
This fact is also observed by Knuth \cite[Page 86, Theorem F]{Knu15} and Harvey and Vondr{\'{a}}k \cite[Corollary 5.28]{HV15} in the more general settings.
Our proof here gives a combinatorial explanation of this equality.

Equation \eqref{eqn:sum-traces} holds whenever $q_{\emptyset}>0$.
Recall that $q_{\emptyset}$ is the shorthand of $q_{\emptyset}(\pvector)$, which is
\begin{align}  \label{eqn:q-empty}
  q_{\emptyset}(\pvector)=\sum_{I\in\mathcal{I}}(-1)^{|I|}\prod_{i\in I}p_i,
\end{align}
where $\mathcal{I}$ is the collection of independent sets of the dependency graph $G$.

\begin{lemma}\label{lem:monotone}
  Assume \Cref{cond:diamond} holds.
  If $q_{\emptyset}(\pvector)>0$, then $q_{\emptyset}(p_1,\dots,p_i z,\dots,p_m)>0$ for any $i\in [m]$ and $0\le z\le 1$.
\end{lemma}
\begin{proof}
  By \eqref{eqn:q-empty},
  \begin{align*}
    q_{\emptyset}(p_1,\dots,p_i z,\dots,p_m) = \sum_{I\in\mathcal{I},\; i\notin I}(-1)^{|I|}\prod_{j\in I}p_j + z \sum_{I\in\mathcal{I},\; i\in I }(-1)^{|I|}\prod_{j\in I}p_j.
  \end{align*}
  Notice that $\sum_{I\in\mathcal{I},\; i\in I }(-1)^{|I|}\prod_{j\in I}p_j=-q_i(\pvector)<0$
  ($q_i(\pvector)$ is the probability of exactly event $A_i$ occurring).
  Hence $q_{\emptyset}(p_1,\dots,p_i z,\dots,p_m)\ge q_{\emptyset}(\pvector)>0$.
\end{proof}

Let $T_i$ be the number of resamplings of event $A_i$ and $T$ be the total number of resampling events.
Then $T=\sum_{i=1}^m T_i$.

\begin{lemma}  \label{lem:expected-i}
   Assume \Cref{cond:diamond} holds.
   If $q_{\emptyset}(\pvector)>0$, then $\Ex T_i=q_{\emptyset}(\pvector)\left( \frac{1}{q_{\emptyset}(p_1,\dots,p_i z,\dots,p_m)} \right)'\bigg \vert_{z=1}$.
\end{lemma}
\begin{proof}
  By Lemma \ref{lem:monotone}, Equation \eqref{eqn:sum-traces} holds with $p_i$ replaced by $p_i z$ where $z\in[0,1]$.
  For a given independent set sequence $\mathcal{S}$,
  let $T_i(\mathcal{S})$ be the total number of occurences of $A_i$ in $\mathcal{S}$.
  Then we have that
  \begin{align}\label{eqn:sum-traces-z}
    \sum_{\mathcal{S}}p_{\mathcal{S}}z^{T_i(\mathcal{S})}=\frac{1}{q_{\emptyset}(p_1,\dots,p_i z,\dots,p_m)}.
  \end{align}
  Take derivative with respect to $z$ of \eqref{eqn:sum-traces-z}:
  \begin{align*}
    \sum_{\mathcal{S}}T_i(\mathcal{S})p_{\mathcal{S}}z^{T_i(\mathcal{S})-1}=\left( \frac{1}{q_{\emptyset}(p_1,\dots,p_i z,\dots,p_m)} \right)'.
  \end{align*}
  Evaluate the equation above at $z=1$:
  \begin{align} \label{eqn:deriv-sum-traces}
    \sum_{\mathcal{S}}T_i(\mathcal{S})p_{\mathcal{S}}=\left( \frac{1}{q_{\emptyset}(p_1,\dots,p_i z,\dots,p_m)} \right)'\bigg \vert_{z=1}.
  \end{align}

  On the other hand, we have that
  \begin{align*} 
    \Ex T_i & =\sum_{\mathcal{S}} \textstyle \PrPRS\left( \text{the log is $\mathcal{S}$} \right) T_i(\mathcal{S})\\
    & = \sum_{\mathcal{S}} p_{\mathcal{S}} q_{\emptyset}(\pvector) T_i(\mathcal{S}) \tag{by Lemma \ref{lem:log-prob}}\\
    & = q_{\emptyset}(\pvector)\left( \frac{1}{q_{\emptyset}(p_1,\dots,p_i z,\dots,p_m)} \right)'\bigg \vert_{z=1}. \tag{by \eqref{eqn:deriv-sum-traces}}
  \end{align*}
  This completes the proof.
\end{proof}

\begin{theorem}
  Assume \Cref{cond:diamond} holds.
  If $q_{\emptyset}>0$,
  then $\Ex T=\sum_{i=1}^{m}\frac{q_{i}}{q_{\emptyset}}$.
  \label{thm:exp-time}
\end{theorem}
\begin{proof}
  Clearly $\Ex T=\sum_{i=1}^m \Ex T_i$.
  By Lemma \ref{lem:expected-i}, all we need to show is that
  \begin{align}\label{eqn:deriv-ratio}
    q_{\emptyset}(\pvector)\left( \frac{1}{q_{\emptyset}(p_1,\dots,p_i z,\dots,p_m)} \right)'\bigg \vert_{z=1} = \frac{q_{i}(\pvector)}{q_{\emptyset}(\pvector)}.
  \end{align}
  This is because 
  \begin{align*}
    q_{\emptyset}'(p_1,\dots,p_i z,\dots,p_m) & = \sum_{i\in J,\; J\in\mathcal{I}}(-1)^{|J|}\prod_{j\in J}p_j\\
    & = - q_{i}(\pvector),
  \end{align*}
  and thus
  \begin{align*}
    \left( \frac{1}{q_{\emptyset}(p_1,\dots,p_i z,\dots,p_m)} \right)' & = \frac{-q_{\emptyset}'(p_1,\dots,p_i z,\dots,p_m)}{q_{\emptyset}(p_1,\dots,p_i z,\dots,p_m)^2}\\
    & = \frac{q_{i}(\pvector)}{q_{\emptyset}(p_1,\dots,p_i z,\dots,p_m)^2}.
  \end{align*}
  It is easy to see that \eqref{eqn:deriv-ratio} follows as we set $z=1$ and the theorem is shown.
\end{proof}

The quantity $\sum_{i=1}^{m}\frac{q_{i}}{q_{\emptyset}}$ is not always easy to bound.
Kolipaka and Szegedy \cite{KS11} have shown that when the probability vector $\pvector$ satisfies Shearer's condition with a constant ``slack'',
the running time is in fact linear in the number of events in the more general setting.
We rewrite it in our notations.


\begin{theorem}[\protect{\cite[Theorem 5]{KS11}}]\label{thm:linear-bound}
  Let $d\ge 2$ be a positive integer and $p_c=\frac{(d-1)^{(d-1)}}{d^d}$.
  Let $p=\max_{i\in[m]}\{p_i\}$.
  If $G$ has maximum degree $d$ and $p < p_c$,
  then $\Ex T \le \frac{p}{p_c-p}\cdot m$.
\end{theorem}

\section{Applications of Algorithm~\ref{ALG:PRS}}

\subsection{Sink-free Orientations}

The goal of this application is to sample a sink-free orientation.
Given a graph $G=(V,E)$, an orientation of edges is a mapping $\sigma$ so that $\sigma(e)=(u,v)$ or $(v,u)$ where $e=(u,v)\in E$.
A \emph{sink} under orientation $\sigma$ is a vertex $v$ so that for any adjacent edge $e=(u,v)$, $\sigma(e)=(u,v)$.
A sink-free orientation is an orientation so that no vertex is a sink.

\prob{Sampling Sink-free Orientations}{A Graph $G$.}{A uniform sink-free orientation.}

The first algorithm for this problem is given by Bubley and Dyer \cite{BD97}, using Markov chains and path coupling techniques.

In this application, we associate with each edge a random variable,
whose possible values are $(u,v)$ or $(v,u)$.
For each vertex $v$, we associate it with a bad event $A_v$,
which happens when $v$ is a sink.
Thus the graph $G$ itself is also the dependency graph.
Condition \ref{cond:diamond} is satisfied, because if a vertex is a sink, then none of its neighbours can be a sink.
Thus we may apply Algorithm \ref{ALG:PRS}, which yields Algorithm \ref{alg:sink-free}.
This is the ``sink-popping'' algorithm given by Cohn, Pemantle, and Propp~\cite{CPP02}.

\begin{algorithm}
  \caption{Sample Sink-free Orientations}
  \label{alg:sink-free}
  \begin{enumerate}[itemsep=0.5em, topsep=0.5em]
    \item Orient each edge independently and uniformly at random.
    \item While there is at least one sink, re-orient all edges that are adjacent to a sink.
    \item Output the current assignment.
  \end{enumerate}
\end{algorithm}

Let $Z_{\mathrm{sink},0}$ be the number of sink-free orientations,
and let $Z_{\mathrm{sink},1}$ be the number of orientations having exactly one sink.
Then Theorem \ref{thm:exp-time} specializes into the following.

\begin{theorem}\label{thm:sink-free}
  The expected number of resampled sinks in Algorithm \ref{alg:sink-free} 
  is $\frac{Z_{\mathrm{sink},1}}{Z_{\mathrm{sink},0}}$.
\end{theorem}

It is easy to see that a graph has a sink-free orientation if and only if the graph is \emph{not} a tree.
The next theorem gives an explicit bound on $\frac{Z_{\mathrm{sink},1}}{Z_{\mathrm{sink},0}}$ when sink-free orientations exist.

\begin{theorem}\label{thm:sink-free:bound}
  Let $G$ be a connected graph on $n$ vertices. 
  If $G$ is not a tree, then $\frac{Z_{\mathrm{sink},1}}{Z_{\mathrm{sink},0}}\le n(n-1)$,
  where $n=\abs{V(G)}$.
\end{theorem}

\begin{proof}
Consider an orientation of the edges of~$G$
with a unique sink at vertex~$v$. 
We give a systematic procedure for transforming this orientation to a sink-free 
orientation.  Since $G$ is connected and not a tree, there exists an
(undirected) path~$\Pi$ in~$G$ of the form $v=v_0,v_1,\ldots,v_{\ell-1},v_\ell=v_k$,
where the vertices $v_0,v_1,\ldots,v_{\ell-1}$ are all distinct and $0\leq k\leq\ell-2$.
In other words, $\Pi$ is a simple path of length $\ell-1$ followed by a single edge back 
to some previously visited vertex.
We will choose a canonical path of this form (depending only on~$G$ and not on the 
current orientation) for each start vertex~$v$.

We now proceed as follows.  Since $v$ is a sink, the first edge on $\Pi$ is directed
$(v_1,v_0)$.  Reverse the orientation of this edge so that it is now oriented 
$(v_0,v_1)$.  This operation destroys the sink at $v=v_0$ but may create a new sink at~$v_1$. 
If $v_1$ is not a sink then halt.  Otherwise, reverse the orientation of the second edge of~$\Pi$
from $(v_2,v_1)$ to $(v_1,v_2)$.  Continue in this fashion:  if we reach $v_i$ and it is not a sink
then halt;  otherwise reverse the orientation of the $(i+1)$th edge from $(v_{i+1},v_i)$
to $(v_i,v_{i+1})$.  This procedure must terminate with a sink-free graph before we reach $v_\ell$.
To see this, note that if we reach the vertex $v_{\ell-1}$ then the final edge of~$\Pi$ must 
be oriented $(v_{\ell-1},v_\ell)$, otherwise the procedure would have terminated already at vertex
$v_k (=v_\ell)$.

The effect of the above procedure is to reverse the orientation of edges on some initial segment $v_0,\ldots,v_i$ of~$\Pi$.  
To put the procedure into reverse, we just need to know the identity of the vertex~$v_i$.  
So our procedure associates at most $n$~orientations having a single sink at vertex~$v$ with each sink-free orientation.
There are $n(n-1)$ choices for the pair $(v,v_i)$, 
and hence $n(n-1)$ single-sink orientations associated with each sink-free orientation.  
This establishes the result.
\end{proof}

\begin{remark}
  The bound in~\Cref{thm:sink-free:bound} is optimal up to a factor of $2$.
  Consider a cycle of length $n$.
  Then there are $2$ sink-free orientations, and $n(n-1)$ single-sink orientations.

 \Cref{thm:sink-free:bound} and~\Cref{thm:sink-free} together yield an $n^2$ bound on the expected number of resamplings that occur during a run of~\Cref{alg:sink-free}.
  A cycle of length $n$ is an interesting special case.
  Consider the number of clockwise oriented edges during a run of the algorithm.
  It is easy to check that this number evolves as an unbiased lazy simple random walk on $[0,n]$.
  Since the walk starts close to $n/2$ with high probability, 
  we know that it will take $\Omega(n^2)$ steps to reach one of the sink-free states, i.e., 0 or~$n$.

  On the other hand, if $G$ is a regular graph of degree $\Delta\geq3$, then we get 
  a much better linear bound from~\Cref{thm:linear-bound}.
  In the case $\Delta=3$, we have $p_c=4/27$, $p=1/8$ and $p/(p_c-p)=27/5$.
  So the expected number of resamplings is bounded by $27n/5$.  
  The constant in the bound improves as $\Delta$ increases.
  Conversely, since the expected running time is exact,
  we can also apply~\Cref{thm:linear-bound} to give an upper bound of $\frac{Z_{\mathrm{sink},1}}{Z_{\mathrm{sink},0}}$ when $G$ is a regular graph.
\end{remark}

\subsection{Rooted Spanning Trees}

Given a graph $G=(V,E)$ with a special vertex $r$,
we want to sample a uniform spanning tree with $r$ as the root.

\prob{Sampling Rooted Spanning Trees}{A Graph $G$ with a vertex $r$.}{A uniform spanning tree rooted at $r$.}

Of course, any given spanning tree may be rooted at any vertex~$r$, so there is no real
difference between rooted and unrooted spanning trees.  However, since this approach 
to sampling generates an oriented tree, it is easier to think 
of the trees as being rooted at a particular vertex~$r$.

For all vertices other than $r$, we randomly assign it to point to one of its neighbours.
This is the random variable associated with~$v$.  We will think of this random 
variable as an arrow $v\to s(v)$ pointing from $v$ to its successor $s(v)$.  The arrows 
point out an oriented subgraph of~$G$ with $n-1$ edges $\{\{v,s(v)\}: v\in V\setminus\{r\}\}$
directed as specified by the arrows.   
The constraint for this subgraph to be a tree rooted at $r$ is that it contains
no directed cycles.  
Note that there are $2^{|E|-|V|+\kappa(G)}$ (undirected) cycles in $G$, 
where $\kappa(G)$ is the number of connected components of $G$.
Hence, we have possibly exponentially many constraints.

Two cycles are dependent if they share at least one vertex.
We claim that Condition \ref{cond:diamond} is satisfied.
Suppose a cycle $C$ is present, and $C'\neq C$ is another cycle that shares at least one vertex with $C$.
If $C'$ is also present,
then we may start from any vertex $v\in C\cap C'$, and then follow the arrows $v\to v'$.
Since both $C$ and $C'$ are present, it must be that $v'\in C\cap C'$ as well.
Continuing this argument we see that $C=C'$.
Contradiction!

As Condition \ref{cond:diamond} is met, we may apply Algorithm \ref{ALG:PRS}, yielding Algorithm \ref{alg:spanning-tree}.
This is exactly the ``cycle-popping'' algorithm by Wilson \cite{Wilson96}, as described in \cite{PW98}.

\begin{algorithm}
  \caption{Sample Rooted Spanning Trees}
  \label{alg:spanning-tree}
  \begin{enumerate}[itemsep=0.5em, topsep=0.5em]
    \item Let $T$ be an empty set. For each vertex $v\neq r$, randomly choose a neighbour $u\in\Gamma(v)$ and add an edge $(v,u)$ to $T$.
    \item While there is at least one cycle in $T$,
      remove all edges in all cycles, and 
      for all vertices whose edges are removed, redo step (1).
    \item Output the current set of edges.
  \end{enumerate}
\end{algorithm}

Let $Z_{\mathrm{tree},0}$ be the number of possible assignments of arrows to the 
vertices of~$G$, that yield a (directed) tree with root~$r$, and 
$Z_{\mathrm{tree},1}$ be the number of assignments that yield 
a unicyclic subgraph.  
The next theorem gives an explicit bound on $\frac{Z_{\mathrm{tree},1}}{Z_{\mathrm{tree},0}}$.

\begin{theorem}\label{thm:spanning-tree:bound}
Suppose $G$ is a connected graph on $n$ vertices, with $m$ edges.
Then $\frac{Z_{\mathrm{tree},1}}{Z_{\mathrm{tree},0}}\le mn$.
\end{theorem}

\begin{proof}
Consider an assignment of arrows to the vertices of $G$ that forms a unicyclic graph. 
This unicyclic subgraph has two components: a directed tree with root~$r$, and a directed cycle with a number of directed subtrees rooted on the cycle.
This is because if we remove the unicyclic component, the rest of the graph has one less edge than vertices and no cycles,
which must be a spanning tree.

As $G$ is connected, there
must be an edge in~$G$ joining the two components;  let this edge be $\{v_0,v_1\}$,
where $v_0$ is in the tree component and $v_1$ in the unicyclic component.  
Now extend this edge to a path $v_0,v_1,\ldots,v_\ell$, by following arrows 
until we reach the cycle.  Thus, $v_1\to v_2,\, v_2\to v_3,\, \ldots,\,v_{\ell-1}\to v_\ell$
are all arrows, and $v_\ell$ is the first vertex that lies on the cycle.  
(It may happen that $\ell=1$.)  Let $v_\ell\to v_{\ell+1}$ be the arrow out 
of $v_\ell$.
Now reassign the arrows from vertices $v_1,\ldots,v_\ell$ thus:  $v_\ell\to v_{\ell-1},\,
\ldots, v_2\to v_1,\, v_1\to v_0$.  Notice that the result is a directed tree rooted at~$r$.

As before, we would like to bound the number of unicyclic subgraphs associated with 
a given tree by this procedure.  We claim that the procedure can be reversed given just
two pieces of information, namely, the edge $\{v_\ell,v_{\ell+1}\}$ and the vertex $v_0$.
Note that, even though the edge $\{v_\ell,v_{\ell+1}\}$ is undirected, we can disambiguate 
the endpoints, as $v_\ell$ is the vertex closer to the root~$r$.  The vertices 
$v_{\ell-1},\ldots,v_1$ are easy to recover, as they are the vertices on the unique 
path in the tree from $v_\ell$ to~$v_0$.  To recover the unicyclic subgraph, we just need to
reassign the arrows for vertices $v_1,\ldots,v_\ell$ as follows:
$v_1\to v_2,\,\ldots,\, v_{\ell}\to v_{\ell+1}$.  

As the procedure can be reversed knowing one edge and one vertex, the number 
of unicyclic graphs associated with each tree can be at most~$mn$.
\end{proof}

Theorem~\ref{thm:spanning-tree:bound} combined with~\Cref{thm:exp-time}
yields an $mn$ upper bound on the expected number of 
``popped cycles'' during a run of~\Cref{alg:spanning-tree}. 

On the other hand, take a cycle of length $n$.
There are $n$ spanning trees with a particular root $v$,
and there are $\Omega(n^3)$ unicyclic graphs (here a cycle has to be of length $2$).
Thus the ratio is $\Omega(n^2)=\Omega(mn)$ since $m=n$,
matching the bound of~\Cref{thm:spanning-tree:bound}.
Moreover, it is known that the cycle-popping algorithm may take $\Omega(n^3)$ time, for example on a dumbbell graph \cite{PW98}.

\subsection{Extremal CNF formulas}\label{sec:ext-CNF}

A classical setting in the study of algorithmic Lov\'asz Local Lemma is to find satisfying assignments in $k$-CNF formulas\footnote{As usual in the study of Lov\'asz Local Lemma, by ``$k$-CNF'' we mean that every clause has exactly size $k$.},
when the number of appearances of every variable is bounded by $d$.
Theorem \ref{thm:LLL} guarantees the existence of a satisfying assignment as long as $d\le \frac{2^k}{ek}+1$.
On the other hand, sampling is apparently harder than searching in this setting.
As shown in \cite[Corollary 30]{BGGGS16}, it is \NP-hard to approximately sample satisfying assignments when $d\ge 5\cdot 2^{k/2}$,
even restricted to the special case of monotone formulas.

Meanwhile, sink-free orientations can be recast in terms of CNF formulas.
Every vertex in the graph is mapped to a clause,
and every edge is a variable.
Thus every variable appears exactly twice,
and we require that the two literals of the same variable are always opposite.
Interpreting an orientation from $u$ to~$v$ as making the literal in the clause corresponding to~$v$ false,
the ``sink-free'' requirement is thus ``not all literals in a clause are false''.
Hence a ``sink-free'' orientation is just a satisfying assignment for the corresponding CNF formula.

To apply Algorithm \ref{ALG:PRS}, we need to require that the CNF formula satisfies Condition \ref{cond:diamond}.
Such formulas are defined as follows.

\begin{definition}\label{def:extremal-CNF}
  We call a CNF formula \emph{extremal} if for every two clauses $C_i$ and $C_j$,
  if there is a common variable shared by $C_i$ and $C_j$,
  then there exists some variable $x$ such that $x$ appears in both $C_i$ and $C_j$ and 
  the two literals are one positive and one negative.
\end{definition}

Let $C_1,\dots,C_m$ be the clauses of a formula $\varphi$.
Then define the bad event $A_i$ as the set of unsatisfying assignments of clause $C_i$.
For an extremal CNF formula, these bad events satisfy Condition \ref{cond:diamond}.
This is because if $A_i\sim A_j$, then by Definition \ref{def:extremal-CNF},
there exists a variable $x\in\vbl(A_i)\cap \vbl(A_j)$ such that 
the unsatisfying assignment of $C_i$ and $C_j$ differ on $x$.
Hence $A_i\cap A_j=\emptyset$.

In this formulation, if the size of $C_i$ is $k$, 
then the corresponding event $A_i$ happens with probability $p_i=\Pr(A_i)=2^{-k}$, 
where variables are sampled uniformly and independently.\footnote{We note that to find a satisfying assignment it is sometimes beneficial to consider non-uniform distributions. See \cite{GST16}.}
Suppose each variable appears at most $d$ times.
Then the maximum degree in the dependency graph is at most $\Delta=(d-1)k$.
Note that in \Cref{thm:linear-bound}, $p_c=\frac{(\Delta-1)^{(\Delta-1)}}{\Delta^\Delta}\ge \frac{1}{e\Delta}$.
Thus if $d\le \frac{2^k}{ek}+1$, then $p_i=2^{-k} < p_c$ 
and we may apply \Cref{thm:linear-bound} to obtain a polynomial time sampling algorithm.

\begin{corollary}\label{cor:ext-CNF}
  For extremal $k$-CNF formulas where each variable appears in at most $d$ clauses,
  if $d\le \frac{2^k}{ek}+1$,
  then Algorithm \ref{ALG:PRS} samples satisfying assignments uniformly at random,
  with $O(m)$ expected resamplings where $m$ is the number of clauses.
\end{corollary}

The condition in Corollary \ref{cor:ext-CNF} essentially matches the condition of Theorem \ref{thm:LLL}.
On the other hand, if we only require Shearer's condition as in Theorem~\ref{thm:Shearer},
\Cref{ALG:PRS} is not necessarily efficient.
More precisely,
let $Z_{\mathrm{CNF},0}$ be the number of satisfying assignments,
and $Z_{\mathrm{CNF},1}$ be the number of assignments satisfying all but one clause.
If we only require Shearer's condition in Theorem \ref{thm:Shearer},
then the expected number of resamplings $\frac{Z_{\mathrm{CNF},1}}{Z_{\mathrm{CNF},0}}$ can be exponential, 
as shown in the next example.

\begin{example}
  Construct an extremal CNF formula $\varphi=C_1\wedge C_2\wedge\dots\wedge C_{4m}$ as follows.
  Let $C_1:=x_1$. Then the variable $x_1$ is pinned to $1$ to satisfy $C_1$.
  Let $C_2:=\overline{x}_1\vee y_1\vee y_2$, $C_3:=\overline{x}_1\vee y_1\vee\overline{y}_2$, and $C_4:=\overline{x}_1\vee\overline{y}_1\vee y_2$.
  Then $y_1$ and $y_2$ are also pinned to $1$ to satisfy all $C_1-C_4$.

  We continue this construction by letting 
  \begin{align*}
    C_{4k+1} & :=\overline{y}_{2k-1}\vee\overline{y}_{2k}\vee x_{k+1},\\
    C_{4k+2} & :=\overline{x}_{k+1}\vee y_{2k+1}\vee y_{2k+2},\\
    C_{4k+3} & :=\overline{x}_{k+1}\vee y_{2k+1}\vee \overline{y}_{2k+2},\\
    C_{4k+4} & :=\overline{x}_{k+1}\vee \overline{y}_{2k+1}\vee y_{2k+2},
  \end{align*}
  for all $1\le k\le m-1$.
  It is easy to see by induction that to satisfy all of them, all $x_i$'s and $y_i$'s have to be $1$.
  Moreover, one can verify that this is indeed an extremal formula.
  Thus $Z_{\mathrm{CNF},0}=1$.
  Moreover, since $\varphi$ has a satisfying assignment and is extremal,
  Shearer's condition is satisfied.
  Note also that $\varphi$ is not a $3$-CNF formula as $C_1$ contains a single variable.

  On the other hand, if we are allowed to ignore $C_1$, then $x_1$ can be $0$.
  In that case, there are $3$ choices of $y_1$ and $y_2$ so that $x_2$ to be $0$ as well.
  Thus, there are at least $3^m$ assignments that only violate $C_1$, where $x_1=x_2=\dots =x_m=0$.
  It implies that $Z_{\mathrm{CNF},1}\ge 3^m$.
  Hence we see that $\frac{Z_{\mathrm{CNF},1}}{Z_{\mathrm{CNF},0}}\ge 3^m$.
  Due to Theorem \ref{thm:exp-time}, the expected running time of Algorithm \ref{ALG:PRS} on this formula $\varphi$ is exponential in $m$.
\end{example}

We will discuss more on sampling satisfying assignments of a $k$-CNF formula in \Cref{sec:k-cnf}.

\section{General Partial Rejection Sampling}\label{sec:G-prs}

In this section we give a general version of Algorithm \ref{ALG:PRS} which can be applied to arbitrary instances in the variable framework,
even without Condition \ref{cond:diamond}.

Recall the notation introduced at the beginning of Section~\ref{sec:PRS}.
So, $\{X_1,\dots,X_n\}$ is a set of random variables,
each with its own distribution and range~$D_i$, and 
$\{A_1,\dots,A_m\}$ is a set of bad events that depend on $X_i$'s.
The dependencies between events are encoded in the dependency graph $G=(V,E)$.
As before, we will use the idea of a resampling table.
Recall that $\val=\val_t=\{\rt{X}{t}\mid 1\le i\le n\}$ denotes the current assignment 
of variables at round $t$, i.e., the elements of the resampling table that 
are active at time~$t$.  Given $\sigma$, 
let $\Bad(\sigma)$ be the set of occurring bad events;
that is, $\Bad(\sigma)=\{i\mid \sigma\in A_i\}$.
For a subset $S\subset V$, let $\partial S$ be the boundary of $S$;
that is, $\partial S=\{i\mid i\notin S\text{ and } \exists j\in S,\; (i,j)\in E\}$.
Moreover, let
\begin{align*}
  \vbl(S):=\bigcup_{i\in S} \vbl(A_i).
\end{align*}
Let $\sigma\vert_S$ be the partial assignment of $\sigma$ restricted to $\vbl(S)$.
For an event $A_i$ and $S\subseteq V$, 
we write $A_i\perp\sigma\vert_S$ if either $\vbl(A_i)\cap\vbl(S) = \emptyset$, 
or there is no way to extend the partial assignment $\sigma\vert_S$ to all variables so that $A_i$ holds.
Otherwise $A_i\not\perp\sigma\vert_S$.

\begin{definition}
  A set $S\subseteq V$ is \emph{unblocking} under $\sigma$ if for every $i\in \partial S$, $A_i\perp\sigma\vert_S$.
\end{definition}

Given $\sigma$, our goal is to resample a set of events that is unblocking and contains $\Bad(\sigma)$.
Such a set must exist because $V$ is unblocking ($\partial V$ is empty) and $\Bad(\sigma)\subseteq V$.
However, we want to resample as few events as possible.

\begin{algorithm}
  \caption{Select the resampling set $\Res(\sigma)$ under an assignment $\sigma$}
  \label{alg:resample}
   $R\gets\Bad(\sigma)$ \tcp*{$R$ is the set of events that will be resampled.}
   $N\gets\emptyset$ \tcp*{$N$ is the set of events that will not be resampled.}
   $U\gets\partial R\setminus N$\;
   \While{$U\neq \emptyset$} 
   {\For{$i\in U$}
   {\uIf{$A_i\not\perp\sigma\vert_R$}{$R\gets R\cup\{i\}$\;}\Else{$N\gets N\cup\{i\}$\;}}
     $U\gets\partial R\setminus N$\;
   }
   \KwRet{$R$}
\end{algorithm}

Intuitively, we start by setting the resampling set $R_0$ as the set of bad events $\Bad(\sigma)$.
We mark resampling events in rounds, similar to a breadth first search.
Let $R_t$ be the resampling set of round $t\ge 0$.
In round $t + 1$, let $A_i$ be an event on the boundary of $R_{t}$ that hasn't been marked yet.
We mark it ``resampling'' if the partial assignments on the shared variables of $A_i$ and $R_{t}$ can be extended so that $A_i$ occurs.
Otherwise we mark it ``not resampling''.
We continue this process until there is no unmarked event left on the boundary of the current $R$.
An event outside of $\Gamma^+(R)$ may be left unmarked at the end of \Cref{alg:resample}.
Note that once we mark some event ``not resampling'', it will never be added into the resampling set.
This is because $R$ is only grow in size during the algorithm.

In \Cref{alg:resample}, we are dynamically updating $R$ during each iteration of going through $U$.
This is potentially beneficial as an event $A_i$ may become incompatible with $R$ after some event $A_j$ is added, where both $i,j\in U$.

We fix a priori an arbitrary ordering while choosing $i\in U$ in the ``for'' loop of Algorithm \ref{alg:resample}.
Then the output of Algorithm \ref{alg:resample} is deterministic under $\sigma$.
Call it $\Res(\sigma)$.

\begin{lemma}  \label{lem:resample-protected}  
  Let $\sigma$ be an assignment.
  For any $i\in\partial \Res(\sigma)$,
  $A_i \perp \sigma\vert_{\Res(\sigma)}$.
\end{lemma}
\begin{proof}
  Since $i\in\partial\Res(\sigma)$, it must have been marked.
  Moreover, $i\not\in \Res(\sigma)$, so it must be marked as ``not resampling''.
  Thus, there exists an intermediate set $R\subseteq\Res(\sigma)$ during the execution of Algorithm \ref{alg:resample} 
  such that $A_i\perp\sigma\vert_R$ and $i\in\partial R$.
  It implies that $A_i$ is disjoint from the partial assignment of $\sigma$ restricted to $\vbl(A_i)\cap\vbl(R)$.
  However, 
  \begin{align*}
    \vbl(A_i)\cap\vbl(R)\subseteq \vbl(A_i)\cap\vbl(\Res(\sigma))
  \end{align*}
  as $R\subseteq\Res(\sigma)$.
  We have that $A_i \perp \sigma\vert_{\Res(\sigma)}$.
\end{proof}

If Condition \ref{cond:diamond} is met,
then $\Res(\sigma)=\Bad(\sigma)$.
This is because at the first step, $R=\Bad(\sigma)$.
By Condition \ref{cond:diamond}, for any $i\in \partial \Bad(\sigma)$,
$A_i$ is disjoint from all $A_j$'s where $j\in \Bad(\sigma)$ and $A_i\sim A_j$.
Since $A_j$ occurs under $\sigma$, $A_i\perp\sigma\vert_{R}$.
Algorithm \ref{alg:resample} halts in the first iteration.
In this case, since the resampling set is just the (independent) set of occurring bad events,
the later Algorithm \ref{ALG:G-PRS} coincides with Algorithm \ref{ALG:PRS}.

The key property of $\Res(\sigma)$ is that if we change the assignment outside of $\Res(\sigma)$,
then $\Res(\sigma)$ does not change, unless the new assignment introduces a new bad event outside of $\Res(\sigma)$.
More formally, we have the following lemma.

\begin{lemma}  \label{lem:resample-select}
  Let $\sigma$ be an assignment.
  Let $\sigma'$ be another assignment such that $\Bad(\sigma') \subseteq \Res(\sigma)$ and 
  such that $\sigma$ and $\sigma'$ agree on all variables in $\vbl(\Res(\sigma))=\bigcup_{i\in \Res(\sigma)}\vbl(A_i)$.
  Then, $\Res(\sigma')=\Res(\sigma)$.
\end{lemma}
\begin{proof}
  Let $R_t(\sigma),N_t(\sigma)$ be the intermediate sets $R,N$, respectively, at time $t$ of the execution of Algorithm \ref{alg:resample} under $\sigma$.
  Thus $R_0(\sigma)=\Bad(\sigma)$ and $R_0(\sigma)\subseteq R_1(\sigma) \subseteq \dots \subseteq \Res(\sigma)$.
  Moreover, $N_0(\sigma)\subseteq N_1(\sigma)\subseteq \cdots$.
  We will show by induction that $R_t(\sigma)=R_t(\sigma')$ and $N_t(\sigma)=N_t(\sigma')$ for any $t\ge 0$.

  For the base case of $t=0$,
  by the condition of the lemma, 
  for every $i\in \Bad(\sigma)\subseteq \Res(\sigma)$, 
  the assignments $\sigma$ and $\sigma'$ agree on $\vbl(A_i)$; or equivalently $\sigma\vert_{\Res(\sigma)}=\sigma'\vert_{\Res(\sigma)}$.
  Together with $\Bad(\sigma')\subseteq \Res(\sigma)$, 
  it implies that $\Bad(\sigma) = \Bad(\sigma')$ and $R_0(\sigma) = R_0(\sigma')$.
  Moreover, $N_0(\sigma)=N_0(\sigma')=\emptyset$.

  For the induction step $t>0$, we have that $R_{t-1}(\sigma)=R_{t-1}(\sigma')\subseteq \Res(\sigma)$ and $N_{t-1}(\sigma)=N_{t-1}(\sigma')$.
  Let $R=R_{t-1}(\sigma)=R_{t-1}(\sigma')$ and $N=N_{t-1}(\sigma)=N_{t-1}(\sigma')$.
  Then we will go through $U=\partial R\setminus N$, which is the same for both $\sigma$ and $\sigma'$.
  Moreover, while marking individual events ``resampling'' or not,
  it is sufficient to look at only $\sigma\vert_R=\sigma'\vert_R$ since $R\subseteq \Res(\sigma)$.
  Thus the markings are exactly the same,
  implying that $R_{t}(\sigma)=R_{t}(\sigma')\subseteq \Res(\sigma)$ and $N_{t}(\sigma)=N_{t}(\sigma')$.
\end{proof}

\begin{algorithm}
  \caption{\gprs}
  \label{ALG:G-PRS}
  \begin{enumerate}[itemsep=0.5em, topsep=0.5em]
    \item Draw independent samples of all variables $X_1,\dots,X_n$ from their respective distributions.
    \item While at least one bad event occurs under the current assignment $\sigma$, 
      use Algorithm \ref{alg:resample} to find $\Res(\sigma)$.
      Resample all variables in $\bigcup_{i\in \Res(\sigma)}\vbl(A_i)$.
    \item When none of the bad events holds, output the current assignment.
  \end{enumerate}
\end{algorithm}

To prove the correctness of Algorithm \ref{ALG:G-PRS}, we will only use three properties of $\Res(\sigma)$,
which are intuitively summarized as follows:
\begin{enumerate}
  \item $\Bad(\sigma)\subseteq\Res(\sigma)$;
  \item For any $i\in\partial \Res(\sigma)$, 
    $A_i$ is disjoint from the partial assignment of $\sigma$ projected on $\vbl(A_i)\cap\vbl(\Res(\sigma))$ (Lemma~\ref{lem:resample-protected});
  \item If we fix the partial assignment of $\sigma$ projected on $\vbl(\Res(\sigma))$, 
    then the output of \Cref{alg:resample} is fixed, unless there are new bad events occurring outside of $\Res(\sigma)$ (Lemma~\ref{lem:resample-select}).
\end{enumerate}

Similarly to the analysis of Algorithm \ref{ALG:PRS},
we call $\mathcal{S}=S_1,\dots,S_{\ell}$ the log, if $S_i$ is the set of resampling events in step $i$ of Algorithm \ref{ALG:G-PRS}.
Note that for Algorithm \ref{ALG:G-PRS},
the log is not necessarily an independent set sequence.
Also, recall that $\val_i$ is the assignment of variables in step $i$,
and $\val_{t}=\val_{T}$ if $T$ is when \Cref{ALG:G-PRS} terminates and $t>T$.
The following lemma is an analogue of \Cref{lem:diamond}.

\begin{lemma}
  Given any log $\mathcal{S}$ of length $\ell\ge 1$,
  $\val_{\ell+1}$ has the product distribution conditioned on none of $A_i$'s occurring where $i\notin \Gamma^{+}(S_\ell)$,
  namely from $\mu\big(\cdot \mid \bigwedge_{i \in [m] \setminus  \Gamma^{+} (S_\ell)} \overline{A_i}\big)$.
  \label{lem:g-prs}
\end{lemma}
\begin{proof}
  Suppose $i\notin \Gamma^+(S_{\ell})$.
  By construction, $S_{\ell}$ contains all occurring bad events of $\sigma_{\ell}$,
  and hence $A_i$ does not occur under $\sigma_{\ell}$.
  In step~$\ell$, we only resample variables that are involved in $S_\ell$,
  so $\sigma_{\ell+1}$ and $\sigma_{\ell}$ agree on $\vbl(A_i)$.
  Hence $A_i$ cannot occur under $\sigma_{\ell+1}$.
  Call an assignment~$\val$ \emph{valid} if none of $A_i$ occurs where $i\notin \Gamma^{+}(S_\ell)$.
  To show that $\val_{\ell+1}$ has the desired conditional product distribution,
  we will show that the probabilities of getting any two valid assignments $\val$ and $\val'$ are 
  proportional to their probabilities of occurrence under the product distrbution $\mu(\cdot)$.

  Let $M$ be the resampling table so that the log of Algorithm \ref{ALG:G-PRS} is $\mathcal{S}$ up to round $\ell$, and $\val_{\ell+1}=\val$.
  Indeed, since we only care about events up to round $\ell+1$, we may truncate the table so that $M=\{X_{i,j} \mid 1\le i\le n,\;\; 1\le j\le j_{i,\ell+1}\}$.
  Let $M'=\{X_{i,j}'\;\vert\; 1\le i\le n,\;\; 1\le j\le j_{i,\ell+1}\}$ be another table 
  where $X_{i,j}'=X_{i,j}$ if $j<j_{i,\ell+1}$ for any $i\in[n]$, and $\val'=(\rt{X'}{\ell+1}:1\leq i\leq n)$
  is a valid assignment.  
  In other words, we only change the last assignment $(\rt{X}{\ell+1}:1\leq i\leq n)$ to another valid assignment.
  We will use $\val'_t=(\rt{X'}{t})$ to denote the active elements of the second resampling table at time~$t$;
  thus $\val'=\val'_{\ell+1}$.

  The lemma follows if Algorithm \ref{ALG:G-PRS} running on $M'$ generates the same log $\mathcal{S}$ up to round $\ell$,
  since, if this is the case, then conditioned on the log $\mathcal{S}$, 
  every possible table $M$ where $\val_{\ell+1}=\val$ is one-to-one correspondence with another table~$M'$ where 
  $\val_{\ell+1}'=\val'$.
  Hence the probability of getting $\val$ is proportional to its weight under $\mu(\cdot)$.

  Suppose otherwise and the log of running Algorithm \ref{ALG:G-PRS} on $M$ and $M'$ differ.
  Let $t_0\le \ell$ be the first round where resampling changes, by which we mean 
  that $\Res(\sigma_{t_0})\neq \Res(\sigma_{t_0}')$.
  By Lemma~\ref{lem:resample-select}, either $\Bad(\sigma_{t_0}')\not\subseteq\Res(\sigma_{t_0})$, 
  or $\sigma_{t_0}\vert_{\Res(\sigma_{t_0})}\neq\sigma_{t_0}'\vert_{\Res(\sigma_{t_0})}$.
  In the latter case, there must be a variable $X_i$ with $i\in\vbl(\Res(\sigma_{t_0}))$ and index $j_{i,\ell+1}$.
  However, $i\in\vbl(\Res(\sigma_{t_0}))$ means that $X_i$ is resampled at least once more in the original run on $M$,
  and its index goes up to at least $j_{i,\ell+1}+1$ at round $\ell+1$.
  A contradiction.
  Thus, $\sigma_{t_0}\vert_{\Res(\sigma_{t_0})}=\sigma_{t_0}'\vert_{\Res(\sigma_{t_0})}$ and $\Bad(\sigma_{t_0}')\not\subseteq\Res(\sigma_{t_0})$.

%
 
  As $\Bad(\sigma_{t_0}')\not\subseteq\Res(\sigma_{t_0})$, 
  there must be a variable $X_{i_0}$ such that $j_{i_0,t_0}=j_{i_0,\ell+1}$ (otherwise $X_{i_0,j_{i_0,t_0}}=X_{i_0,j_{i_0,t_0}}'$)
  and an event $A_k$ such that $i_0\in \vbl(A_k)$, $k\in \Bad(\sigma_{t_0}')$ but $k\not \in \Res(\sigma_{t_0})$.
  Suppose first that $\forall i\in\vbl(A_k)$, $j_{i,t_0}=j_{i,\ell+1}$, 
  which means that all variables of $A_k$ have reached their final values in the $M$ run at time $t_0$.
  This implies that $k\notin \Gamma^+(S_{t})$ for any $t\ge t_0$ as otherwise some of the variables in $\vbl(A_k)$ would be resampled at least once after round $t_0$.
  In particular, $k\notin \Gamma^+(S_{\ell})$.
  This contradicts with $\sigma'$ being valid.

  Otherwise there are some variables in $\vbl(A_k)$ that get resampled after time $t_0$ in the $M$ run.
  Let $t_1$ be the first such time and $Y\subset \vbl(A_k)$ be the set of variables resampled at round $t_1$;
  namely, $Y=\vbl(A_k)\cap\vbl(\Res(\sigma_{t_1}))$.
  We have that $\sigma_{t_1}\vert_{Y}=\sigma_{t_0}\vert_{Y}$ because $t_1$ is the first time of resampling variables in $Y$.      
  Moreover, as variables of $Y$ have not reached their final values yet in the $M$ run,
  $\sigma_{t_0}\vert_{Y}=\sigma_{t_0}'\vert_{Y}$.
  Thus, $\sigma_{t_1}\vert_{Y}=\sigma_{t_0}'\vert_{Y}$.

  Assuming $k\in \Res(\sigma_{t_1})$ 
  would contradict the fact that $X_{i_0}$ has reached its final value in the $M$ run.
  Hence $k\notin \Res(\sigma_{t_1})$, but nevertheless variables in $Y\subset\vbl(A_k)$ are resampled.
  This implies that $k\in\partial \Res(\sigma_{t_1})$.
  By Lemma \ref{lem:resample-protected}, $A_k\perp \sigma_{t_1}\vert_{\Res(\sigma_{t_1})}$.
  As $\vbl(A_k)$ cannot be disjoint from $\vbl(\Res(\sigma_{t_1}))$,
  this means that $A_k$ is imcompartible with the partial assignment of $\sigma_{t_1}$ restricted to $\vbl(A_k)\cap\vbl(\Res(\sigma_{t_1}))=Y$.
  Equivalently, $A_k\perp \sigma_{t_1}\vert_{Y}$.
  However we know that $\sigma_{t_1}\vert_{Y}=\sigma_{t_0}'\vert_{Y}$,
  so $A_k\perp \sigma_{t_0}'\vert_{Y}$, contradicting $k\in \Bad(\sigma_{t_0}')$.
\end{proof}

\begin{theorem}
  If Algorithm \ref{ALG:G-PRS} halts, 
  then its output has the product distribution conditioned on none of $A_i$'s occurring.
  \label{thm:g-prs}
\end{theorem}

\begin{proof}
%
%
  Let a sequence $\mathcal{S}$ of sets of events be the log of any successful run. 
  Then $S_{\ell}=\emptyset$.
  By Lemma \ref{lem:g-prs}, conditioned on the log $\mathcal{S}$, the output assignment $\sigma$ 
  is $\mu\big(\cdot \mid \bigwedge_{i \in [m] \setminus  \Gamma^{+} (S_\ell)} \overline{A_i}\big)
  = \mu\big(\cdot \mid \bigwedge_{i \in [m]} \overline{A_i}\big)$.
  This is valid for any possible log,
  and the theorem follows.
\end{proof}

\section{Running Time Analysis of Algorithm~\ref{ALG:G-PRS}}


Obviously when there is no assignment avoiding all bad events,
then Algorithm \ref{ALG:G-PRS} will never halt.
Thus we want to assume some conditions to guarantee a desired assignment.
However, the optimal condition of Theorem \ref{thm:Shearer} is quite difficult to work under.
Instead, in this section we will be working under the assumption that the asymmetric LLL condition \eqref{eqn:LLL} holds.
In fact, to make the presentation clean, we will mostly work with the simpler symmetric case.

However, as mentioned in Section \ref{sec:ext-CNF},
\cite[Corollary 30]{BGGGS16} showed that even under the asymmetric LLL condition~\eqref{eqn:LLL},
sampling can still be \NP-hard.
We thus in turn look for further conditions to make Algorithm \ref{ALG:G-PRS} efficient.

Recall that $\mu(\cdot)$ is the product distribution of sampling all variables independently.
For two distinct events $A_i\sim A_j$, let $R_{ij}$ be the event that the partial assignments on $\vbl(A_i)\cap\vbl(A_j)$ can be extended to an assignment making $A_j$ true.
Thus, if $A_i\not\perp\sigma\vert_S$ for some event set $S$,
then $R_{ji}$ must hold for all $A_j\in S$ and $A_j\sim A_i$.
Conversely, it is possible that each individual $R_{ji}$ is true for all $A_j\in R$ and $A_j\sim A_i$,
and yet $A_i\perp\sigma\vert_S$.
Also note that $R_{ij}$ is not necessarily the same as $R_{ji}$.
Let $r_{ij}:=\mu(R_{ij})$.

Define $\displaystyle p:=\max_{i\in[m]}p_i$ and $\displaystyle r:=\max_{A_i\sim A_j,\;i\neq j} r_{ij}$.
Let $\Delta$ be the maximum degree of the dependency graph $G$.
The main result of the section is the following theorem.

\begin{theorem}\label{thm:exp-time:G-prs}
  Let $m$ be the number of events and $n$ be the number of variables.
  For any $\Delta\ge 2$, if $\Cone ep\Delta^2 \le 1$ and $\Ctwo er\Delta \le 1$, 
  then the expected number of resampled events of Algorithm \ref{ALG:G-PRS} is $O(m)$.

  Moreover, when these conditions hold, the number of rounds is $O(\log m)$
  and the number of variable resamples is $O(n\log m)$, both in expectation and with high probability.  
\end{theorem}

The first condition $\Cone ep\Delta^2 \le 1$ is stronger than the condition of the symmetric Lov\'asz Local Lemma,
but this seems necessary since \cite[Corollary 30]{BGGGS16} implies that if $p\Delta^2\ge C$ for some constant $C$ then the sampling problem is \NP-hard.
Intuitively, the second condition $\Ctwo er\Delta \le 1$ bounds the expansion from bad events to resampling events at every step of Algorithm \ref{ALG:G-PRS}.
We will prove Theorem \ref{thm:exp-time:G-prs} in the rest of the section.

Let $S\subseteq [m]$ be a subset of vertices of the dependency graph $G$.
Recall that $A(S)$ is the event $\bigwedge_{i\in S}A_i$ and $B(S)$ is the event $\bigwedge_{i\in S}\overline{A_i}$.
Moreover, $S^c$ is the complement of $S$, namely $S^c=[m]\setminus S$,
and $S^e$ is the ``exterior'' of $S$, namely $S^e=[m] \setminus \Gamma^+(S)$.

Lemma \ref{lem:g-prs} implies that if we resample $S$ at some step $t$ of Algorithm \ref{ALG:G-PRS},
then at step $t+1$ the distribution is the product measure $\mu$ conditioned on none of the events in the exterior of $S$ holds; namely $\Pr_\mu(\cdot\mid B(S^e))$.

Let $E$ be an event (not necessarily one of $A_i$) depending on a set $\vbl(E)$ of variables.
Let $\Gamma(E):=\{i\mid i\in[m],\;\vbl(A_i)\cap\vbl(E)\neq\emptyset\}$ if $E$ is not one of $A_i$,
and $\Gamma(A_i):=\{j\mid j\in[m],\;j\not=i\text{ and }\vbl(A_j)\cap\vbl(A_i)\neq\emptyset\}$ is defined as usual.
Let $S\subseteq [m]$ be a subset of vertices of $G$.
The next lemma bounds the probability of $E$ conditioned on none of the events in $S$ happening.
It was first observed in \cite{HSS11}.
We include a proof for completeness (which is a simple adaption of the ordinary local lemma proof).

\begin{lemma}[\protect{\cite[Theorem 2.1]{HSS11}}]\label{lem:event-bound}
  Suppose \eqref{eqn:LLL} holds.
  For an event $E$ and any set $S\subseteq [m]$,
  \begin{align*}
    \Pr_\mu(E\mid B(S))\le \Pr_\mu(E)\prod_{i\in \Gamma(E)\cap S}(1-x_i)^{-1},
  \end{align*}
  where $x_i$'s are from \eqref{eqn:LLL}.
\end{lemma}
\begin{proof}
  We prove the inequality by induction on the size of $S$.
  The base case is when $S$ is empty and the lemma holds trivially.
  
  For the induction step, let $S_1=S\cap\Gamma(E)$ and $S_2=S\setminus S_1$.
  If $S_1=\emptyset$, then the lemma holds trivially as $E$ is independent from $S$ in this case.
  Otherwise $S_2$ is a proper subset of $S$.
  We have that
  \begin{align*}
    \Pr_\mu(E\mid B(S)) & = \frac{\Pr_\mu(E\wedge B(S_1) \mid B(S_2))}{\Pr_\mu(B(S_1)\mid B(S_2))} \\
    & \le \frac{\Pr_\mu(E \mid B(S_2))}{\Pr_\mu(B(S_1)\mid B(S_2))}\\
    & = \frac{\Pr_\mu(E)}{\Pr_\mu(B(S_1)\mid B(S_2))},
  \end{align*}
  where the last line is because $E$ is independent from $B(S_2)$.
  We then use the induction hypothesis to bound the denominator.
  Suppose $S_1=\{j_1,j_2,\dots,j_r\}$ for some $r>0$.
  Then,
  \begin{align*}
    \Pr_\mu(B(S_1)\mid B(S_2)) & = \Pr_\mu\left(\bigwedge_{i\in S_1}\overline{A_i} \; \Bigg| \bigwedge_{i\in S_2}\overline{A_i}\right)\\
    & = \prod_{t=1}^r\Pr_\mu\left(\overline{A_{j_t}}\; \Bigg| \bigwedge_{s=1}^{t-1}\overline{A_{j_s}}\wedge\bigwedge_{i\in S_2}\overline{A_i}\right) \\
    & = \prod_{t=1}^r\left( 1- \Pr_\mu\left(A_{j_t}\; \Bigg| \bigwedge_{s=1}^{t-1}\overline{A_{j_s}}\wedge\bigwedge_{i\in S_2}\overline{A_i}\right) \right).  
  \end{align*}
  By the induction hypothesis and \eqref{eqn:LLL},
  we have that for any $1\le t\le r$,
  \begin{align*}
    \Pr_\mu\left(A_{j_t}\; \Bigg| \bigwedge_{s=1}^{t-1}\overline{A_{j_s}}\wedge\bigwedge_{i\in S_2}\overline{A_i}\right) 
    & \le \Pr_\mu(A_{j_t}) \prod_{i\in \Gamma(j_t)}(1-x_i)^{-1} \\
    & \le x_{j_t} \prod_{i\in \Gamma(j_t)}(1-x_i) \prod_{i\in \Gamma(j_t)}(1-x_i)^{-1} \\
    & = x_{j_t}.
  \end{align*}
  Thus, 
  \begin{align*}
    \Pr_\mu(B(S_1)\mid B(S_2)) & \ge \prod_{i\in S_1}\left( 1- x_i \right).
  \end{align*}
  The lemma follows.
\end{proof}

Typically we set $x_i=\frac{1}{\Delta+1}$ in the symmetric setting.
Then \eqref{eqn:LLL} holds if $ep(\Delta+1)\le 1$.
In this setting, Lemma~\ref{lem:event-bound} is specialized into the following.
\begin{corollary}\label{cor:symmetric-bound}
  If $ep(\Delta+1)\le 1$, then 
  \begin{align*}
    \Pr_\mu(E\mid B(S))\le \Pr_\mu(E)\left( 1+ \frac{1}{\Delta} \right)^{\abs{\Gamma(E)}}.
  \end{align*}
\end{corollary}

In particular, if $e p (\Delta+1)\le 1$, for any event $A_i$ where $i\not\in S$, by Corollary \ref{cor:symmetric-bound},
\begin{align}
  \Pr_\mu(A_i\mid B(S)) &\le p_i \left( \frac{\Delta+1}{\Delta} \right)^\Delta \le ep. \label{eqn:R-bound}
\end{align}
%

Let $\Res_t$ be the resampling set of Algorithm \ref{ALG:G-PRS} at round $t\ge 1$,
and let $\Bad_t$ be the set of bad events present at round $t$.
If \Cref{ALG:G-PRS} has already stopped at round $t$, then $\Res_t=\Bad_t=\emptyset$.
Furthermore, let $\Bad_0=\Res_0=[m]$ since in the first step all random variables are fresh.

Let $C:=1-p$.

\begin{lemma}\label{lem:g-prs-Res-shrink}
  For any $\Delta\ge 2$, 
  if $\Cone ep\Delta^2 \le 1$ and $\Ctwo er\Delta\le 1$,
  then,
  \begin{align*}
    \Ex\left(\abs{\Res_{t+1}} \mid \Res_0,\dots,\Res_t\right) \le C \abs{\Res_t}.
  \end{align*}  
\end{lemma}
\begin{proof}
  Clearly for any $\Delta\ge 2$,
  the condition $\Cone ep\Delta^2\le 1$ implies that $ep(\Delta+1)\le 1$.
  Therefore the prerequisite of Corollary \ref{cor:symmetric-bound} is met.
  Notice that, by Lemma~\ref{lem:g-prs}, we have that
  \begin{align*}
    \Ex\left(\abs{\Res_{t+1}} \mid \Res_0,\dots,\Res_t\right) = \Ex\left(\abs{\Res_{t+1}} \mid \Res_t\right).
  \end{align*}
  We will show in the following that 
  \begin{align*}
    \Ex\big(\abs{\Res_{t+1}}\bigm| \text{the set of resampling events at round $t$ is (exactly) $\Res_t$}\big)\le C \abs{\Res_t},
  \end{align*}
  where $C=1-p$.
  This implies the lemma.


  Call a path $i_0,i_1,\dots,i_{\ell}$ where $\ell\ge 0$ in the dependency graph $G$ \emph{bad} if the following holds:
  \begin{enumerate}
    \item $i_0\in\Bad_{t+1}$;
    \item the event $R_{i_{k-1}i_{k}}$ holds for every $1\le k\le \ell$;
    \item any $i_k$ ($k\in[\ell]$) is not adjacent to $i_{k'}$ unless $k'=k-1$ or $k+1$.
  \end{enumerate}
  Indeed, paths having the third property are induced paths in $G$.
  If $i\in\Res_{t+1}$, $A_i$ must be added by \Cref{alg:resample} during some iteration of the while loop.
  In the $0$th iteration, all of $\Bad_{t+1}$ are added.
  We claim that for any $i\in \Res_{t+1}$ added in $\ell\ge 0$ iteration by~\Cref{alg:resample}, 
  there exists at least one bad path such that $i_0,i_1,\dots,i_{\ell}=i$.
  We show the claim by an induction on $\ell$.
  \begin{itemize}
    \item The base case is that $\ell=0$, and thus $i\in\Bad_{t+1}$.
      The bad path is simply $i$ itself.
    \item For the induction step $\ell\ge 1$, due to~\Cref{alg:resample}, 
      there must exist $i_{\ell-1}$ adjacent to $i_\ell=i$ 
      such that $i_{\ell-1}$ has been marked ``resampling'' during iteration $\ell-1$, 
      and $R_{i_{\ell-1}i_\ell}$ occurs.
      By the induction hypothesis, there exists a bad path $i_0,\dots,i_{\ell-1}$.
      Since $i$ is not marked at iteration $\ell-1$, 
      $i$ is not adjacent to any vertices that has been marked up to iteration $\ell-2$.
      Thus $i_{\ell}$ is not adjacent to any $i_k$ where $k\le \ell-2$,
      and the path $i_0,\dots,i_{\ell-1},i_{\ell}$ is bad.
  \end{itemize}

  We next turn to bounding the number of bad paths.
  It is straightforward to bound the size of $\Bad_{t+1}\subseteq \Gamma^+(\Res_t)$.
  If $i\in \Bad_{t+1}$, then there are two possibilities.
  The first scenario is that $i\in\Res_t$ and then all of its random variables are fresh.
  In this case it occurs with probability $p_i\le p$.
  Otherwise $i\in\partial\Res_t$.
  Recall that by Lemma~\ref{lem:g-prs},
  the distribution at round $t+1$ is $\Pr_{\mu}(\cdot\mid B(\Res_t^e))$.
  By Corollary \ref{cor:symmetric-bound}, for any $i\in\partial\Res_t$,
  \begin{align*}
    \Pr_{\mu}\left(A_i\mid B(\Res_t^e)\right) \le p\left( 1+\frac{1}{\Delta} \right)^\Delta \le ep.
  \end{align*}
  This implies that 
  \begin{align}
    & \Ex\left(\abs{\Bad_{t+1}}\mid \text{the set of resampling events at round $t$ is (exactly) $\Res_t$}\right)\notag\\
    \le &\; p\abs{\Res_t}+ ep\abs{\partial\Res_t} \le p(1+e\Delta)\abs{\Res_t}.\label{eqn:Bad-t+1}
  \end{align}
  
  Next we bound the size of $\abs{\Res_{t+1}\setminus\Bad_{t+1}}$.
  Let $P=i_0,\dots,i_{\ell}$ be an induced path; 
  that is, for any $k\in[\ell]$, $i_k$ is not adjacent to $i_{k'}$ unless $k'=k-1$ or $k+1$.
  Only induced paths are potentially bad.
  Moreover, $P$ contributes to $\abs{\Res_{t+1}\setminus\Bad_{t+1}}$ only if its length $\ell\ge 1$.
  Let $D_P$ be the event that $P$ is bad.
  In other words, $D_P:= A_{i_0}\wedge R_{i_0i_1}\wedge \dots\wedge R_{i_{\ell-1}i_{\ell}}$.
  By Lemma~\ref{lem:g-prs}, we have that
  \begin{align}\label{eqn:D_P-rewrite}
    &\; \Pr(P\text{ is bad at round $t+1$}\mid \text{the set of resampling events at round $t$ is $\Res_t$}) \notag\\
    = &\; \Pr_\mu (D_P \mid B(\Res_t^e)),
  \end{align}
  where we recall that we denote $\Res_t^e = \inb{m} \setminus \Gamma^+(\Res_t)$.
  Applying Corollary \ref{cor:symmetric-bound} with $S=\Res_t^e$, we have that
  \begin{align}\label{eqn:D_P-condition-bound}
    \Pr_\mu (D_P \mid B(\Res_t^e)) \le \Pr_\mu(D_P) \left( 1+\frac{1}{\Delta} \right)^{\abs{\Gamma(D_P)}}.
  \end{align}
  Note that $\Gamma(R_{i_ki_{k+1}})\subseteq \Gamma^+(A_{i_k})$ for all $0\le k\le \ell-1$.
  By the definition of $D_P$,
  \begin{align*}
    \Gamma(D_P) & \subseteq \Gamma(A_{i_0})\cup \Gamma(R_{i_0i_1})\cup\dots\cup\Gamma(R_{i_{\ell-1}i_{\ell}}) \\
    & \subseteq \Gamma^+(A_{i_0})\cup \Gamma^+(A_{i_1})\cup\dots\cup\Gamma^+(A_{i_{\ell-1}}),
  \end{align*}
  implying that 
  \begin{align}\label{eqn:D_P-neighbour-bound}
    \abs{\Gamma(D_P)}  & \le \ell(\Delta+1),
  \end{align}
  as $\abs{\Gamma^+(A_k)}\le \Delta+1$ for all $0\le k\le\ell-1$.

  We claim that $A_{i_0}$ is independent from $R_{i_{k-1}i_{k}}$ for any $2\le k\le \ell$.
  This is because $i_k$ is not adjacent to $i_0$ for any $k\ge 2$, implying that
  \begin{align*}
    \vbl(R_{i_{k-1}i_k})\cap \vbl(A_{i_0}) & = \vbl(A_{i_{k-1}})\cap\vbl(A_{i_{k}})\cap\vbl(A_{i_0})\\
    & \subseteq \vbl(A_{i_{k}})\cap\vbl(A_{i_0}) = \emptyset.    
  \end{align*}
  Moreover, any two events $R_{i_{k-1}i_k}$ and $R_{i_{k'-1}i_{k'}}$ are independent of each other as long as $k < k'$.
  This is also due to the third property of bad paths.
  Since $k < k'$, we see that $\abs{k'-(k-1)}\ge 2$ and $i_{k'}$ is not adjacent to $i_{k-1}$.
  It implies that 
  \begin{align*}
    \vbl(R_{i_{k-1}i_k})\cap \vbl(R_{i_{k'-1}i_{k'}}) & = \vbl(A_{i_{k-1}})\cap\vbl(A_{i_{k}})\cap\vbl(A_{i_{k'-1}})\cap\vbl(A_{i_{k'}}) \\
    & \subseteq \vbl(A_{i_{k-1}})\cap\vbl(A_{i_{k'}}) = \emptyset.
  \end{align*}
  The consequence of these independences is 
  \begin{align}
    \Pr_\mu(D_P) & \le \Pr_\mu(A_{i_0} \wedge R_{i_1i_2}\wedge \dots\wedge R_{i_{\ell-1}i_{\ell}})\notag\\
    & = \Pr_\mu(A_{i_0})\prod_{k=2}^{\ell} \Pr_\mu(R_{i_{k-1}i_k})\notag\\
    & \le p r^{\ell-1}. \label{eqn:D_P-Pr-bound}
  \end{align}
  Note that in the calculation above we ignore $R_{i_0i_1}$ as it can be positively correlated to $A_{i_0}$.

  Combining \eqref{eqn:D_P-rewrite}, \eqref{eqn:D_P-condition-bound}, \eqref{eqn:D_P-neighbour-bound}, and \eqref{eqn:D_P-Pr-bound},
  we have that
  \begin{align} \label{eqn:D_P-bound}
    & \Pr(D_P\mid \text{the set of resampling events at round $t$ is (exactly) $\Res_t$}) \notag \\
    \le &\; pr^{\ell-1}\left( 1+\frac{1}{\Delta} \right)^{\ell(\Delta+1)} \le \frac pr \left( \left( 1+\frac{1}{\Delta} \right)e r \right)^{\ell}.
  \end{align}

  In order to apply a union bound on all bad paths, we need to bound their number.
  The first vertex $i_0$ must be in $\Bad_{t+1}$, implying that $i_0\in \Gamma^+(\Res_t)$.
  Hence there are at most $(\Delta+1)\abs{\Res_t}$ choices.
  Then there are at most $\Delta$ choices of $i_1$ and $(\Delta-1)$ choices of every subsequent $i_{k}$ where $k\ge 2$.
  Hence, there are at most $\Delta(\Delta-1)^{\ell-1}$ induced paths 
  of length $\ell\ge 1$, originating from a particular $i_0\in \Gamma^+(\Res_t)$.
  Thus, by a union bound on all potentially bad paths and \eqref{eqn:D_P-bound},
  \begin{align}
    & \Ex\big(\abs{\Res_{t+1}\setminus\Bad_{t+1}}\bigm| \text{the set of resampling events at round $t$ is (exactly) $\Res_t$}\big)\notag\\
    \le &\; \sum_{\ell=1}^{\infty} (\Delta+1)\abs{\Res_t}\Delta(\Delta-1)^{\ell-1}p/r \left( \left( 1+\frac{1}{\Delta} \right)e r \right)^{\ell}\notag\\
    = & \; \frac{(\Delta+1)\Delta p}{(\Delta-1)r}\abs{\Res_t}\sum_{\ell=1}^{\infty} \left( \left( \frac{\Delta^2-1}{\Delta} \right) e r \right)^\ell\notag\\
    \le & \; \frac{(\Delta+1)\Delta p}{(\Delta-1)r}\abs{\Res_t}\sum_{\ell=1}^{\infty} \left( e r \Delta \right)^\ell 
    = \frac{(\Delta+1)\Delta p}{(\Delta-1)r}\cdot\frac{e r \Delta}{1- e r \Delta} \abs{\Res_t}\notag\\
    \le & \; \frac{\Delta+1}{\Delta-1}\cdot \frac{3}{2}\cdot ep\Delta^2\abs{\Res_t},\label{eqn:Res-Bad-t+1}
  \end{align}
  where we use the condition that $ e r \Delta \le 1/\Ctwo$.

  Combining \eqref{eqn:Bad-t+1} and \eqref{eqn:Res-Bad-t+1}, we have that 
  \begin{align*}
    & \Ex\left(\abs{\Res_{t+1}}\mid \text{the set of resampling events at round $t$ is (exactly) $\Res_t$}\right)\notag\\
    \le &\; \frac{\Delta+1}{\Delta-1}\cdot \frac{3}{2}\cdot ep\Delta^2\abs{\Res_t}+  p(1+e\Delta)\abs{\Res_t}\\
    = & \; p\left(\frac{\Delta+1}{\Delta-1}\cdot \frac{3}{2}\cdot e\Delta^2 +  (1+e\Delta)\right)\abs{\Res_t}.
  \end{align*}
  All that is left is to verify that
  \begin{align*}
    p\left(\frac{\Delta+1}{\Delta-1}\cdot \frac{3}{2}\cdot e\Delta^2 +  (1+e\Delta)\right) \le C,
  \end{align*}
  where $C=1-p$.
  This is straightforward by the condition $\Cone ep\Delta^2\le 1$ and $\Delta\ge 2$, as 
  \begin{align*}
    C - p\left(\frac{\Delta+1}{\Delta-1}\cdot \frac{3}{2}\cdot e\Delta^2 +  (1+e\Delta)\right) 
    & \ge \Cone ep\Delta^2 -p - p\left(\frac{\Delta+1}{\Delta-1}\cdot \frac{3}{2}\cdot e\Delta^2 +  (1+e\Delta)\right)\\
    & \ge p\left(\Cone e\Delta^2 - 1 - \frac{\Delta+1}{\Delta-1}\cdot \frac{3}{2}\cdot e\Delta^2 -  (1+e\Delta) \right) \ge 0.\qedhere
  \end{align*}
\end{proof}

For $t\ge 1$, by \Cref{lem:g-prs-Res-shrink} and the law of iterated expectations,
\begin{align*}
  \Ex \abs{\Res_t}\le C \Ex \abs{R_{t-1}}.
\end{align*}
Thus, $\Ex \abs{\Res_t}\le C^t \abs{R_0} = C^t m$.
As $C<1$, the expected number of resampling events is
\begin{align*}
  \sum_{t=0}^\infty \Ex \abs{\Res_t} \le \sum_{t=0}^\infty C^t m = \frac{1}{1-C}\cdot m.
\end{align*}
This implies the first part of \Cref{thm:exp-time:G-prs}.
For the second part, just observe that after $O(\log m)$ rounds, 
the expected number of bad events is less than $m^{-c}$ for any constant $c$, and Markov inequality applies.

The first condition of \Cref{thm:exp-time:G-prs} requires $p$ to be roughly $O(\Delta^{-2})$.
This is necessary, due to the hardness result in \cite{BGGGS16} (see also Theorem \ref{thm:sharing-CNF-hard}).
Also, in the analysis, it is possible to always add all of $\partial\Bad_t$ into $\Res_t$.
Consider a monotone CNF formula.
If a clause is unsatisfied, then all of its neighbours need to be added into the resampling set.
Such behaviours would eventually lead to the $O(\Delta^{-2})$ bound.
This situation is in contrast to the resampling algorithm of Moser and Tardos \cite{MT10},
which only requires $p=O(\Delta^{-1})$ as in the symmetric Lov\'asz Local Lemma.

Also, we note that monotone CNF formulas, in which all correlations are positive, seem to be the worst instances for our algorithms.
In particular, \Cref{ALG:G-PRS} is exponentially slow when the underlying hypergraph of the monotone CNF is a (hyper-)tree.
This indicates that our condition on $r$ in \Cref{thm:exp-time:G-prs} is necessary for \Cref{ALG:G-PRS}.
In contrast, Hermon et al.\ \cite{HSZ16} show that on a linear hypergraph (including the hypertree), 
the Markov chain mixes rapidly for degrees higher than the general bound.
It is unclear how to combine the advantages from these two approaches.

\section{Applications of Algorithm~\ref{ALG:G-PRS}}

\subsection{\texorpdfstring{$k$}{k}-CNF Formulas} \label{sec:k-cnf}

Consider a $k$-CNF formula where every variable appears in at most $d$ clauses.
Then Theorem \ref{thm:LLL} says that if $d\le 2^k/(e k)+1$, then there exists a satisfying assignment.
However, as mentioned in Section \ref{sec:ext-CNF},
\cite[Corollary 30]{BGGGS16} showed that when $d\ge 5\cdot 2^{k/2}$,
then sampling satisfying assignments is \NP-hard, even restricted to monotone formulas.

To apply Algorithm \ref{ALG:G-PRS} in this setting, we need to bound the parameter $r$ in Theorem \ref{thm:exp-time:G-prs}.
A natural way is to lower bound the number of shared variables between any two \emph{dependent} clauses.
If this lower bound is $s$, then $r=2^{-s}$ since there is a unique assignment on these $s$ variables 
that can be extended in such a way as to falsify the clauses.

\newcommand{\sharingCNF}[2]{\textsc{Deg}-#1-\textsc{Sha}-#2-\textsc{CNF}}
\begin{definition}
  Let $d\ge 2$ and $s\ge 1$.
  A $k$-CNF formula is said to have \emph{degree} $d$ if every variable appears in at most $d$ clauses.
  Moreover, it has \emph{intersection} $s$ if for any two clauses $C_i$ and $C_j$ that share at least one variable,
  $\abs{\vbl(C_i)\cap\vbl(C_j)}\ge s$.
\end{definition}

Note that by the definition if $k<s$ then the formula is simply isolated clauses.
Otherwise, $k\ge s$ and we have that $p_i = p = 2^{-k}$ and $r\le 2^{-s}$.
A simple double counting argument indicates that the maximum degree $\Delta$ in the dependency graph satisfies $\Delta\le \frac{dk}{s}$.

We claim that for integers $d$ and $k$ such that $d\ge 3$ and $dk\ge 2^{3e}$,
conditions $d\le \frac{2^{k/2}}{6e}$ and $s\ge \min\{\log_2 dk,k/2\}$ imply the conditions of \Cref{thm:exp-time:G-prs}, 
namely, $6ep\Delta^2\le 1$ and $3er\Delta \le 1$.
In fact, if $s\ge \log_2 dk \ge \log_2 d$, then
\begin{align*}
  6ep\Delta^2 \le 6e2^{-k} \left( \frac{dk}{s} \right)^2 
  \le 6e2^{-k} \left( \frac{dk}{\log_2 d} \right)^2 
  \le 6e \left( \frac{k}{6e\left( k/2 - \log_2 6e \right)} \right)^2 < 1,
\end{align*}
as $\frac{d}{\log_2 d}$ is increasing for any $d\ge 3$.
Moreover,
\begin{align*}
  3er\Delta \le \frac{3e dk}{2^s s} \le \frac{3e}{\log_2 (dk)} \le 1.
\end{align*}
Otherwise $k/2\le s\le \log_2 dk$, which implies that
\begin{align*}
  6ep\Delta^2 \le 6e2^{-k} \left( \frac{dk}{s} \right)^2 
  \le 6e2^{-k} \left( \frac{dk}{k/2} \right)^2 
  \le 6e 2^{-k} \left(\frac{2^{k/2}}{3e}\right)^2 < 1,
\end{align*}
and
\begin{align*}
  3er\Delta \le \frac{3e dk}{2^s s} \le \frac{6e d k}{k 2^{k/2}} \le 1.
\end{align*}
Thus by Theorem \ref{thm:exp-time:G-prs} we have the following result.
Note that resampling a clause involves at most $k$ variables,
and for $k$-CNF formulas with degree $d$,
the number of clauses is linear in the number of variables.

\begin{corollary}  \label{cor:sharing-CNF}
  For integers $d$ and $k$ such that $d\ge 3$ and $dk\ge 2^{3e}$,
  if $ d\le \frac{1}{6e}\cdot 2^{k/2}$ and $s\ge \min\{\log_2 dk,k/2\}$,
  then Algorithm \ref{ALG:G-PRS} samples satisfying assignments of $k$-CNF formulas with degree $d$ and intersection $s$
  in $O(n)$ time in expectation and in $O(n\log n)$ time with high probability,
  where $n$ is the number of variables.
\end{corollary}

We remark that the lower bound on intersection size $s$ in Corollary \ref{cor:sharing-CNF} does not make the problem trivial.
Note that the lower bound $\min\{\log_2 dk,k/2\}$ is at most $k/2$.
The ``hard'' instance in the proof of \cite[Corollary 30]{BGGGS16} has roughly $k/2$ shared variables for each pair of dependent clauses.
For completeness, we will show that if $k$ is even, and $d\ge 4\cdot 2^{k/2}$ and $s=k/2$, then the sampling problem is \NP-hard.
The proof is almost identical to that of \cite[Corollary 30]{BGGGS16}.
The case of odd $k$ can be similarly handled but with larger constants.

We will use the inapproximability result of Sly and Sun \cite{SS14} (or equivalently, of Galanis et al.\ \cite{GSV16}) for the hard-core model. 
We first remind the reader of the relevant definitions. 
Let $\lambda>0$. 
For a graph $G=(V,E)$, the hard-core model with parameter $\lambda>0$ is a probability distribution over the set of independent sets of $G$; 
each independent set $I$ of $G$ has weight proportional to $\lambda^{|I|}$. 
The normalizing factor of this distribution is the partition function $Z_G(\lambda)$, 
formally defined as $Z_G(\lambda):=\sum_{I} \lambda^{|I|}$ where the sum ranges over all independent sets $I$ of $G$.
The hardness result we are going to use is about approximating $Z_G(\lambda)$,
but it is standard to translate it into the sampling setting as the problem is self-reducible.

\begin{theorem}[\cite{SS14,GSV16}]\label{thm:slysunhardcore}
  For $d\geq 3$, 
  let $\lambda_c(d):=(d-1)^{d-1}/(d-2)^{d}$. 
  For all $\lambda>\lambda_c(d)$, it is \NP-hard to sample an independent set $I$ with probability proportional to $\lambda^{|I|}$ in a $d$-regular graph.
\end{theorem}

\begin{theorem}\label{thm:sharing-CNF-hard}
  Let $k$ be an even integer.
  If $d\ge 4\cdot 2^{k/2}$ and $s=k/2$,
  then it is \NP-hard to sample satisfying assignments of $k$-CNF formulas with degree $d$ and intersection $s$ uniformly at random.
\end{theorem}
\begin{proof}
  Given a $d$-regular graph $G=(V,E)$,
  we will construct a monotone $k$-CNF formula $C$ with degree $d$ and intersection $k/2$
  such that satisfying assignments of $C$ can be mapped to independent sets of $G$.
  Replace each vertex $v\in V$ by $s$ variables, say $v_1,\dots,v_s$.
  If $(u,v)\in E$, then create a monotone clause $v_1\vee\dots\vee v_s\vee u_1\vee\dots\vee u_s$.
  It is easy to see that every variable appears exactly $d$ times since $G$ is $d$-regular.
  Moreover, the number of shared variables is always $s$ and the clause size is $2s=k$.

  For each satisfying assignment, we map it to a subset of vertices of $G$.
  If all of $v_1,\dots,v_s$ are \emph{false}, then make $v$ occupied.
  Otherwise $v$ is unoccupied.
  Thus a satisfying assignment is mapped to an independent set of $G$.
  Moreover, there are $(2^{k/2}-1)^{n-\abs{I}}$ satisfying assignments corresponding to an independent set~$I$,
  where $n$ is the number of vertices in $G$.
  Thus the weight of $I$ is proportional to $(2^{k/2}-1)^{-\abs{I}}$; namely $\lambda=(2^{k/2}-1)^{-1}$ in the hard-core model.

  In order to apply Theorem \ref{thm:slysunhardcore}, all we need to do is to verify that $\lambda>\lambda_c$, or equivalently
  \begin{align*}
    2^{k/2}-1<\frac{(d-2)^{d}}{(d-1)^{d-1}}.
  \end{align*}
  This can be done as follows,
  \begin{align*}
    \frac{(d-2)^{d}}{(d-1)^{d-1}}& =(d-2)\Big(1-\frac{1}{d-1}\Big)^{d-1}\geq \bigg(\frac{4}{5}\bigg)^5(d-2)> 2^{k/2}-1.\qedhere
  \end{align*}
\end{proof}

Due to Theorem \ref{thm:sharing-CNF-hard},
we see that the dependence between $k$ and $d$ in Corollary \ref{cor:sharing-CNF} is tight in the exponent,
even with the further assumption on intersection $s$.

\subsection{Independent Sets}

We may also apply Algorithm~\ref{ALG:G-PRS} to sample hard-core configurations with parameter $\lambda$.
Every vertex is associated with a random variable which is occupied with probability $\frac{\lambda}{1+\lambda}$.
In this case, each edge defines a bad event which holds if both endpoints are occupied.
Thus $p=\left( \frac{\lambda}{1+\lambda} \right)^2$.
Algorithm \ref{ALG:G-PRS} is specialized to Algorithm \ref{alg:Hardcore}.

\begin{algorithm}[htbp]
  \caption{Sample Hard-core Configurations}
  \label{alg:Hardcore}
  \begin{enumerate}[itemsep=0.5em, topsep=0.5em]
    \item Mark each vertex occupied with probability $\frac{\lambda}{1+\lambda}$ independently.
    \item While there is at least one edge with both end points occupied,
      resample all occupied components of sizes at least $2$ and their boundaries.
    \item Output the set of vertices.
  \end{enumerate}
\end{algorithm}

\newcommand\BadVtx{ {\sf BadVtx}}
\newcommand\ResVtx{ {\sf ResVtx}}

To see this, consider a graph $G=(V,E)$ with maximum degree $d$.
Given a configuration $\sigma:V\rightarrow\{0,1\}$, consider the subgraph $G[\sigma]$
of $G$ induced by the vertex subset $\{v\in V:\sigma(v)=1\}$.  Then 
we denote by $\BadVtx(\sigma)$ the set of vertices in any component of $G[\sigma]$ of size at least $2$.
Then the output of Algorithm \ref{alg:resample} is
\begin{align*}
  \ResVtx(\sigma) := \BadVtx(\sigma) \cup \partial\BadVtx(\sigma).
\end{align*}
This is because first, all of $\partial\BadVtx(\sigma)$ will be resampled,
since any of them has at least one occupied neighbour in $\BadVtx(\sigma)$.
Secondly, $v\in \partial\BadVtx(\sigma)$ is unoccupied (otherwise $v\in\BadVtx(\sigma)$), 
and \Cref{alg:resample} stops after adding all of $\partial\BadVtx(\sigma)$.
This explains \Cref{alg:Hardcore}.

Moreover, let $\Bad(\sigma)$ be the set of edges whose both endpoints are occupied under $\sigma$.
Let $\Res(\sigma)$ be the set of edges whose both endpoints are in $\ResVtx(\sigma)$.
Let $\sigma_t$ be the random configuration of Algorithm \ref{alg:Hardcore} at round $t$ if it has not halted,
and $\Bad_t=\Bad(\sigma_t)$, $\Res_t=\Res(\sigma_t)$.

\begin{lemma}  \label{lem:hardcore-shrink}
  If $ep(2d-1)<1$,
  then $\Ex\abs{\Bad_{t+1}} \le (4ed^2-1)p \Ex\abs{\Bad_{t}}$.
\end{lemma}
\begin{proof}
  First note that the dependency graph is the line graph of $G$ and $\Delta=2d-2$ is the maximum degree of the line graph of $G$.
  Thus $ep(2d-1)<1$ guarantees the prerequisite of Corollary \ref{cor:symmetric-bound} is met.
  It also implies that for any $\sigma$, $\abs{\Res(\sigma)}\le (2d-1)\Bad(\sigma)$,
  and $\partial\abs{\Res(\sigma)}\le (2d-2)\abs{\Res(\sigma)}$.
  Similarly to the analysis in Lemma \ref{lem:g-prs-Res-shrink}, conditioned on a fixed $\Bad_t$, 
  by Corollary \ref{cor:symmetric-bound} (or \eqref{eqn:R-bound} in particular), we have that
  \begin{align*}
    \Ex \abs{\Bad_{t+1}} &\le p\abs{\Res(\sigma)}+ep\abs{\partial\Res(\sigma)}\\
    &\le \left(p(2d-1)+ep(2d-2)(2d-1)\right)\abs{\Bad_t}\\
    &< (4ed^2-1)p \abs{\Bad_t}.
  \end{align*}
  Since the inequality above holds for any $\Bad_t$,
  the lemma follows.
\end{proof}

Lemma~\ref{lem:hardcore-shrink} implies that, if $4epd^2\le 1$, then the number of bad edges shrinks with a constant factor,
and~\Cref{alg:Hardcore} resamples $O(m)$ edges in expectation and $O(m\log m)$ edges with high probability,
where $m=\abs{E}$.
A bounded degree graph is sparse and thus $m=O(n)$, where $n$ is the number of vertices.
Since $p=\left( \frac{\lambda}{1+\lambda} \right)^2$, the condition $4epd^2\le 1$ is equivalent to
\begin{align*}
  \lambda \le \frac{1}{2\sqrt{e}d-1}.
\end{align*}
Thus we have the following theorem, where the constants are slightly better than directly applying Theorem \ref{thm:exp-time:G-prs}.

\begin{theorem}\label{thm:hard-core}
  If $\lambda \le \frac{1}{2\sqrt{e}d-1}$, 
  then Algorithm \ref{alg:Hardcore} draws a uniform hard-core configuration with parameter $\lambda$ from a graph with maximum degree $d$ 
  in $O(n)$ time in expectation and in $O(n\log n)$ time with high probability,
  where $n$ is the number of vertices.
\end{theorem}

The optimal bound of sampling hard-core configurations is $\lambda<\lambda_c\approx\frac{e}{d}$ where $\lambda_c$ is defined in~\Cref{thm:slysunhardcore}.
The algorithm is due to Weitz \cite{Wei06} and the hardness is shown in \cite{SS14,GSV16}.
The condition of our~\Cref{thm:hard-core} is more restricted than correlation decay based algorithms \cite{Wei06} or traditional Markov chain based algorithms.
Nevertheless, our algorithm matches the correct order of magnitude $\lambda=O(d^{-1})$.
Moreover, our algorithm has the advantage of being simple, exact, and running in linear time in expectation.

\newcommand{\sigmahat}{\hat\sigma}
\newcommand{\what}{\widehat w}
\newcommand{\What}{\widehat W}
\newcommand{\Zhat}{\widehat Z}
\newcommand{\muhat}{\hat\mu}

\section{Distributed algorithms for sampling}
An interesting feature of Algorithm~\ref{ALG:G-PRS} is that it is distributed.\footnote{See \cite{FSY17} for a very recent work by Feng, Sun, and Yin on distributed sampling algorithms.
In particular, they show a similar lower bound in \cite[Section 5]{FSY17}.}
For concreteness, 
consider the algorithm applied to sampling hard-core configurations on a graph~$G$ (i.e.\ \Cref{alg:Hardcore}), 
assumed to be of bounded maximum degree.  
Imagine that each vertex is assigned a processor that has access to a source of random bits.  
Communication is possible between adjacent processors and is assumed to take constant time. 
This is essentially Linial's LOCAL model \cite{Lin87}.
Then, in each parallel round of the algorithm, the processor at vertex~$v$ can update the value $\sigma(v)$ in constant time, 
as this requires access only to the values of $\sigma(u)$ for vertices $u\in V(G)$ within a bounded distance~$r$ of $v$.  
In the case of the hard-core model, we have $r=2$, since the value~$\sigma(v)$ at vertex~$v$ should be updated precisely 
if there are vertices $u$ and $u'$ such that $v\sim u$ 
and $u\sim u'$ and $\sigma(u)=\sigma(u')=1$.  Note that we allow $u'=v$ here.

In certain applications, including the hard-core model, Algorithm~\ref{ALG:G-PRS} runs in 
a number of rounds that is bounded by a logarithmic function of the input size with high probability.
(Recall \Cref{thm:exp-time:G-prs}.)
We show that this is optimal.  (Although the argument is presented in the
context of the hard-core model, it ought to generalise to many other applications.) 

Set $L=\lceil c\log n\rceil$ for some constant $c>0$ to be chosen later.
The instance that establishes the lower bound is a graph~$G$ consisting of 
a collection of $n/L$ disjoint paths $\Pi_1,\ldots,\Pi_{n/L}$ with $L$~vertices each.  
(Assume that $n$ is an exact multiple of $L$; this is not a significant restriction.) 
The high-level idea behind the lower bound is simple, and consists of two observations. 
We assume first that the distributed algorithm we are considering always 
produces an output, say $\sigmahat:V(G)\to\{0,1\}$, within $t$~rounds.  
It will be easy at the end to extend
the argument to the situation where the running time is a possibly unbounded random 
variable with bounded expectation.  

Focus attention on a particular path~$\Pi$ with endpoints $u$ and~$v$.  
The first observation 
is that if $rt<L/2$ then $\sigma(u)$ (respectively, $\sigma(v)$) 
depends only on the computations performed by processors in 
the half of~$\Pi$ containing~$u$ (respectively~$v$).  Therefore, in the algorithm's 
output, $\sigmahat(u)$ and $\sigmahat(v)$ are probabilistically independent.
The second observation is that if the constant~$c$ is sufficiently small then, in the
hard-core distribution, $\sigma(u)$ and $\sigma(v)$ are significantly correlated.
Since the algorithm operates independently on each of the $n/L$ paths, these 
small but significant correlations combine to force to a large variation
distance between the hard-core distribution and the output distribution 
of the algorithm.

We now quantify the second observation.
Let $\sigma:V(G)\to\{0,1\}$ be a sample from the hard-core distribution on a path~$\Pi$ 
on $k$~vertices with endpoints $u$ and~$v$, 
and let $I_k=Z_\Pi(\lambda)$ denote the corresponding hard-core partition function 
(weighted sum over independent sets). 
Define the matrix $W_k=\big(\begin{smallmatrix}w_{00}&w_{01}\\w_{10}&w_{11}\end{smallmatrix}\big)$, 
where $w_{ij}=\Pr(\sigma(u)=i\wedge\sigma(v)=j)$.
Then
$$
W_k=\frac1{I_k}\begin{pmatrix}I_{k-2}&\lambda I_{k-3}\\\lambda I_{k-3}&\lambda^2I_{k-4}\end{pmatrix},
$$
since $I_k$ is the total weight of independent sets in~$\Pi$, 
$I_{k-2}$ is the total weight of independent sets
with $\sigma(u)=\sigma(v)=0$, $I_{k-3}$ is the total weight of independent sets
with $\sigma(u)=0$ and $\sigma(v)=1$, and so on.
Also note that $I_k$ satisfies the recurrence
\begin{equation}\label{eq:Ikrecurrence}
I_0=1,\quad I_1=\lambda+1,\quad\text{and}\quad 
I_k=I_{k-1}+\lambda I_{k-2},\>\text{for $k\geq2$}.
\end{equation}

We will use $\det W_k$ to measure the deviation of  
the distribution of $(\sigma(u),\sigma(v))$ from a product distribution.
Write 
$$W_k'=\begin{pmatrix}I_{k-2}& I_{k-3}\\ I_{k-3}&I_{k-4}\end{pmatrix},$$
and note that $\det W_k=\lambda^2I_k^{-2}\det W_k'$.
Applying recurrence (\ref{eq:Ikrecurrence}) once to each of the four entries of $W_k'$, we have
\begin{align*}
\det W_k'&=I_{k-2}I_{k-4}-I_{k-3}^2\\
&=(I_{k-3}+\lambda I_{k-4})(I_{k-5}+\lambda I_{k-6})-(I_{k-4}+\lambda I_{k-5})^2\\
&=I_{k-3}(I_{k-5}+\lambda I_{k-6})-I_{k-4}(I_{k-4}+\lambda I_{k-5})+
\lambda^2(I_{k-4}I_{k-6}-I_{k-5}^2)\\
&=I_{k-3}I_{k-4}-I_{k-4}I_{k-3}+\lambda^2\det W_{k-2}'\\
&=\lambda^2\det W_{k-2}',
\end{align*}
for all $k\geq 6$.  By direct calculation, $\det W_4'=-\lambda^2$ and $\det W_5'=\lambda^3$.  Hence, by induction, 
$\det W_k'=(-1)^{k-1}\lambda^{k-2}$, and 
\begin{equation}\label{eq:Wk}
\det W_k=\frac{(-1)^{k-1}\lambda^k}{I_k^2},
\end{equation}
for all $k\geq4$.

Solving the recurrence (\ref{eq:Ikrecurrence}) gives the following formula for~$I_k$: 
$$
I_k=A_\lambda\left(\frac{1+\sqrt{4\lambda+1}}2\,\right)^k + B_\lambda\left(\frac{1-\sqrt{4\lambda+1}}2\,\right)^k,
$$
where 
$$
A_\lambda=\left(\frac12+\frac{2\lambda+1}{2\sqrt{4\lambda+1}}\right)\quad \text{and}\quad
B_\lambda=\left(\frac12-\frac{2\lambda+1}{2\sqrt{4\lambda+1}}\right)
$$
Asymptotically,
$$
I_k=(1+o(1))A_\lambda\left(\frac{1+\sqrt{4\lambda+1}}2\,\right)^k.
$$
Substituting this estimate into (\ref{eq:Wk}) yields $|\det W_k|=(1+o(1))A_\lambda^{-2}\alpha^k$ where 
$$\alpha=\frac{2\lambda}{2\lambda+\sqrt{4\lambda+1}+1}.$$
Note that $0<\alpha<1$ and $\alpha$ depends only on~$\lambda$.

Now let the matrix
$\What_k=\big(\begin{smallmatrix}\what_{00}&\what_{01}\\ \what_{10}&\what_{11}\end{smallmatrix}\big)$
be defined as for $W_k$, but with respect to the output distribution of the distributed 
sampling algorithm rather than the true hard-core distribution.  
Recall that we choose $L=\lceil c\log n\rceil>2rt$, which implies that $\sigmahat(u)$ and $\sigmahat(v)$ are independent and $\det\What_L=0$.
It is easy to check that if $\|\What_k-W_k\|_\infty\leq\varepsilon$, where the matrix norm is 
entrywise, then $|\det W_k|\leq\varepsilon$.
Thus, for $c$ sufficiently small (and $L=\lceil c\log n\rceil$), we can ensure that 
$\|\What_L-W_L\|_\infty\geq n^{-1/3}$.  Thus, $|\what_{ij}-w_{ij}|\geq n^{-1/3}$,
for some $i,j$;
for definiteness, suppose that $i=j=0$ and that $\what_{00}>w_{00}$. 

Let $Z$ (respectively $\Zhat$) be the number of paths whose endpoints are both assigned~0 
in the hard-core distribution (respectively, the algorithm's output distribution).
Then $Z$ (respectively $\Zhat$) is a binomial random variable with expectation $\mu=w_{00}n/L$
(respectively $\muhat=\what_{00}n/L$).  Since $|\Ex Z-\Ex \Zhat|>\Omega(n^{2/3}/\log n)$, 
a Chernoff bound gives that $\Pr(Z\geq (\mu+\muhat)/2)$ and $\Pr(\Zhat\leq (\mu+\muhat)/2)$
both tend to zero exponentially fast with $n$.  It follows that the variation distance
between the distributions of $\sigma$ and~$\sigmahat$ is $1-o(1)$.  

The above argument assumes an absolute bound on running time, whereas the running time 
of an exact sampling algorithm will in general be a random variable~$T$.  To bridge the 
gap, suppose $\Pr(T\leq t)\geq\frac23$.  Then
\begin{align*}
  \|\sigmahat-\sigma\|_\mathrm{TV}&=\max_A\bigl|\Pr(\sigmahat\in A)-\Pr(\sigma\in A)\bigr|\\
  &=\max_A\Bigl|\big(\Pr(\sigmahat\in A\mid T\leq t)-\Pr(\sigma\in A)\big)\Pr(T\leq t)\\
  &\qquad\qquad\null+\bigl(\Pr(\sigmahat\in A\mid T> t)-\Pr(\sigma\in A)\big)\Pr(T>t)\Bigr|\\
  &\geq \tfrac23(1-o(1))-\tfrac13\times1,
\end{align*}
Where $\|\cdot\|_\mathrm{TV}$ denotes variation distance, and $A$ ranges over
events $A\subseteq \{0,1\}^{|V(G)|}$.
Thus $\|\sigma-\sigmahat\|_\mathrm{TV}\geq\frac13-o(1)$, which is a contradiction.
It follows that $\Pr(T\leq t)<\frac23$ and hence 
$\Ex(T)\geq\frac13t$.  Note that this argument places a lower bound on
parallel time not just for exact samplers, but even for (very) approximate ones.

With only a slight increase in work, one could take the instance $G$ to be a path 
of length~$n$, which might be considered more natural.  Identify $O(n/L)$ subpaths
within~$G$, suitably spaced, and of length~$L$.  The only complication is that 
the hard-core distribution does not have independent marginals on distinct subpaths.
However, by ensuring that the subpaths are separated by distance $n^{\alpha}$, 
for some small $\alpha>0$, the correlations can be controlled, 
and the argument proceeds, with only slight modification, as before.

\section*{Acknowledgements}

We would like to thank Yumeng Zhang for pointing out a factor $k$ saving in Corollary \ref{cor:sharing-CNF}.
We thank Dimitris Achlioptas, Fotis Iliopoulos, Pinyan Lu, Alistair Sinclair, and Yitong Yin for their helpful comments.
We also thank anonymous reviewers for their detailed comments.

HG and MJ are supported by the EPSRC grant EP/N004221/1. 
JL is supported by NSF grant CCF-1420934.
This work was done (in part) while the authors were visiting the Simons Institute for the Theory of Computing.
HG was also supported by a Google research fellowship in the Simons Institute.

\bibliographystyle{plain}
\bibliography{PRS}

\end{document}